\documentclass{article}

\usepackage[margin=2cm]{geometry}

\usepackage[T1]{fontenc}
\usepackage{libertine}
\usepackage[italic]{mathastext}

\usepackage{amsmath,amssymb,amsfonts,amsthm}
\usepackage{tikz}
\usetikzlibrary{matrix}

\usepackage{subcaption} 

\usepackage{titling} 
\usepackage[numbers, sort, comma, square]{natbib}

\usepackage[most]{tcolorbox}
\tcbuselibrary{breakable}
\usepackage{enumitem}

\DeclareMathOperator*{\argmax}{arg\,max}
\DeclareMathOperator*{\argmin}{arg\,min}
\newcommand{\eqdef}{\mathbin{\stackrel{\rm def}{=}}}
\makeatletter
\def\hlinewd#1{
	\noalign{\ifnum0=`}\fi\hrule \@height #1 \futurelet
	\reserved@a\@xhline}
\makeatother

\newtheorem{theorem}{Theorem}[section]
\newtheorem{lemma}[theorem]{Lemma}
\newtheorem{claim}[theorem]{Claim}

\theoremstyle{definition}

\newtheorem{definition}{Definition}[section]

\theoremstyle{remark}

\newcommand{\R}{\mathbb{R}}
\newcommand{\E}{\mathbb{E}}
\newcommand{\pP}{\mathcal{P}}
\newcommand{\dD}{\mathcal{D}}
\newcommand{\iI}{\mathcal{I}}

\newcommand{\bv}[1]{\mathbf{#1}}
\newcommand{\mean}{\mathrm{mean}}
\newcommand{\diag}{\mathrm{diag}}
\newcommand{\rsum}{\mathrm{rowsum}}

  \usepackage{nth}
  \usepackage{intcalc}

\begin{document}
	
	\title{Understanding Filter Bubbles and Polarization \\in Social Networks}

\author{Uthsav Chitra\\Princeton University\\ \texttt{\small uchitra@cs.princeton.edu} \and Christopher Musco\\Princeton University\\ \texttt{\small cmusco@cs.princeton.edu}}

	\maketitle
	
		\begin{abstract}
		Recent studies suggest that social media usage --- while linked to an increased diversity of information and perspectives for users --- has exacerbated user polarization on many issues.
		A popular theory for this phenomenon centers on the concept of ``filter bubbles": by automatically recommending content that a user is likely to agree with, social network algorithms create echo chambers of similarly-minded users that would not have arisen otherwise \cite{Pariser:2011}.
		However, while echo chambers have been observed in real-world networks, the evidence for filter bubbles is largely post-hoc.
		
		In this work, we develop a mathematical framework to study the filter bubble theory. 
		We modify the classic Friedkin-Johnsen opinion dynamics model by introducing another actor, the \emph{network administrator}, who filters content for users by making small changes to the edge weights of a social network (for example, adjusting a news feed algorithm to change the level of interaction between users).
		
		On real-world networks from Reddit and Twitter, we show that when the network administrator is incentivized to reduce disagreement among users, even relatively small edge changes can result in the formation of echo chambers in the network and increase user polarization. We theoretically support this observed sensitivity of social networks to outside intervention by analyzing synthetic graphs generated from the \emph{stochastic block model}. Finally, we show that a slight modification to the incentives of the network administrator can mitigate the filter bubble effect while minimally affecting the administrator's target objective, user disagreement.
	\end{abstract}
	
	\begin{figure}
		\begin{subfigure}[t]{0.31\textwidth}
			\includegraphics[width=.95\linewidth,height=3.75cm]{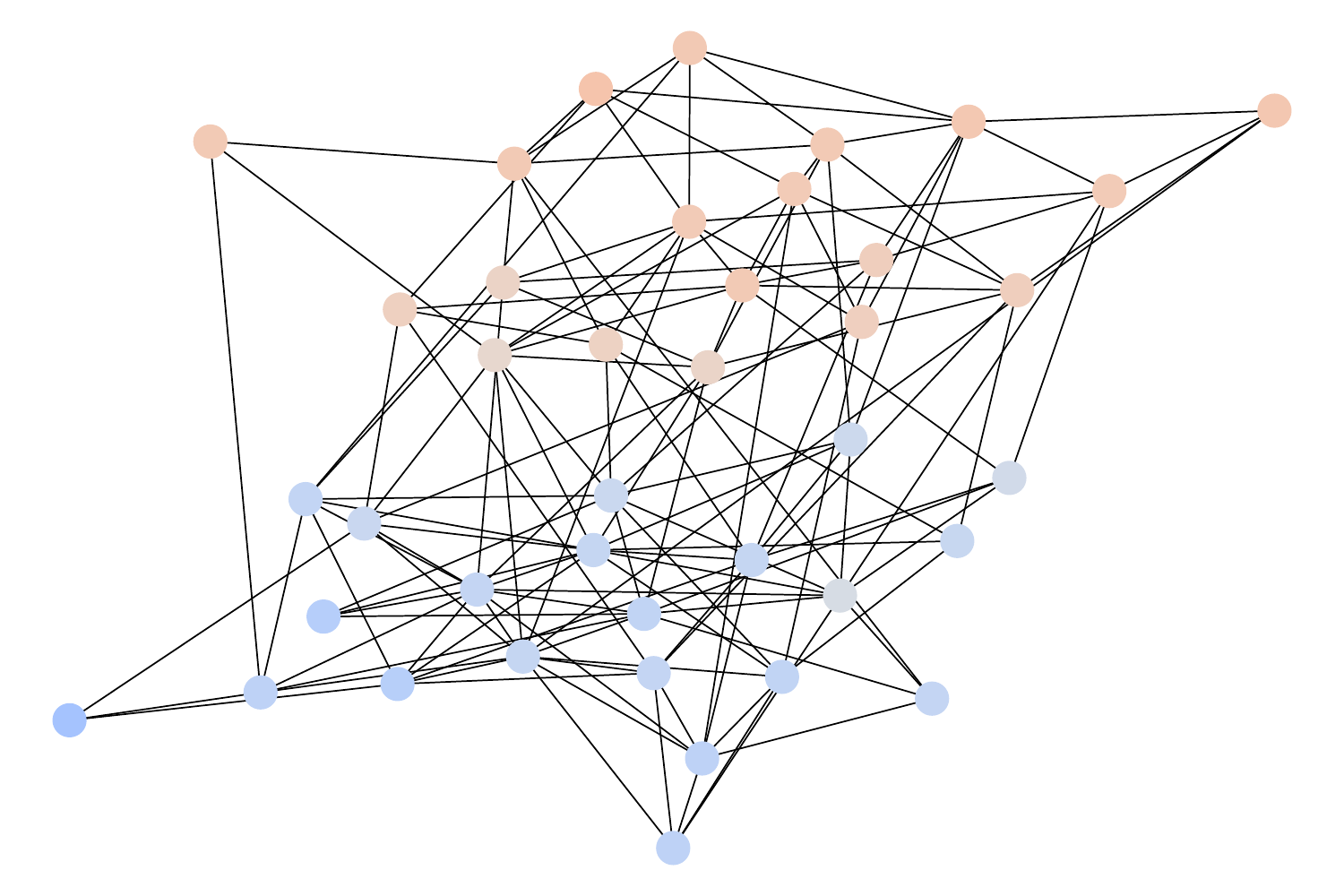}
			\caption{Example synthetic social network graph.} \label{fig:1a}
		\end{subfigure}
		\hfill 
		\begin{subfigure}[t]{0.31\textwidth}
			\includegraphics[width=.95\linewidth,height=3.75cm]{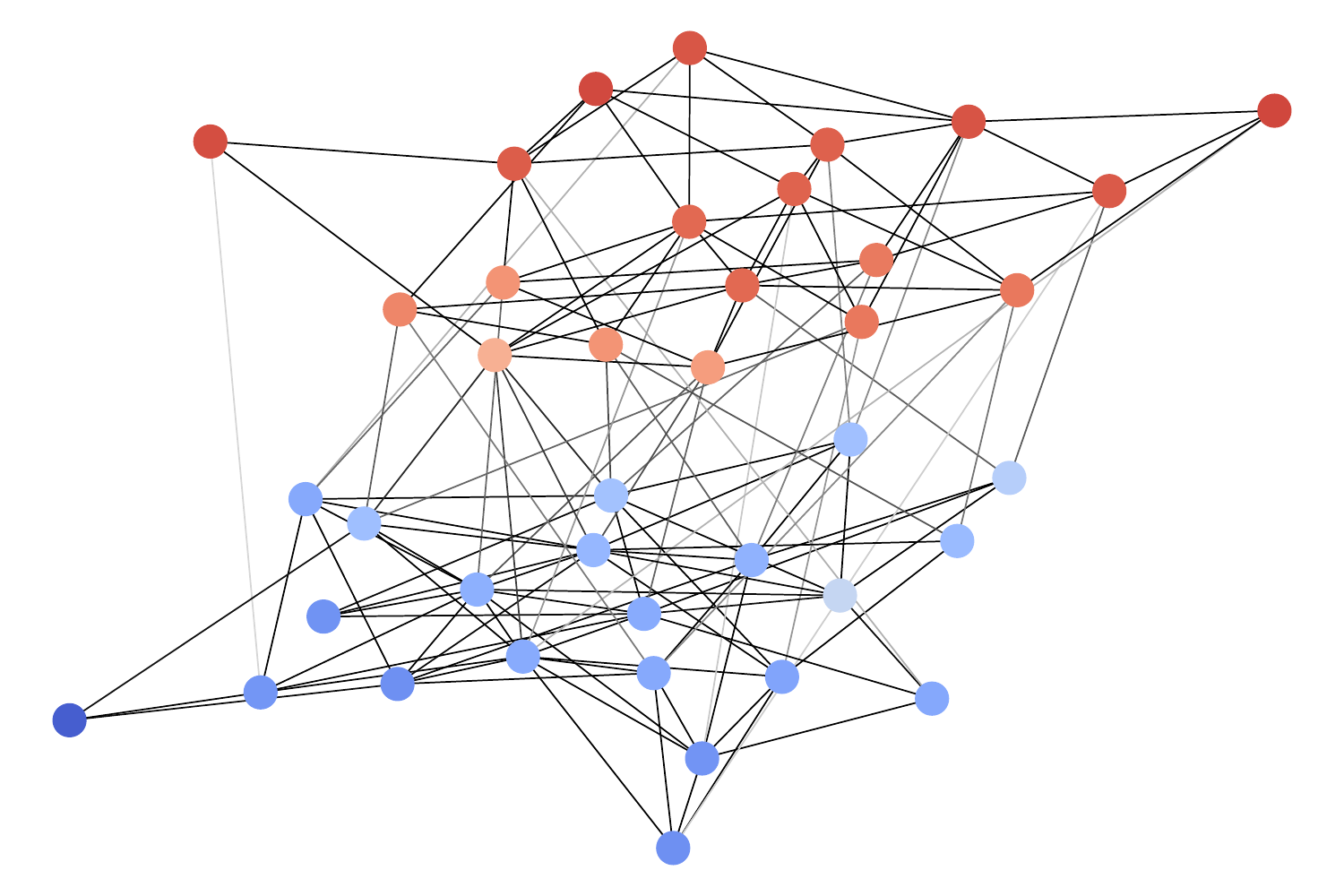}
			\caption{Graph after network administrator changes just 20\% of edge weight.} \label{fig:1b}
		\end{subfigure}
		\hfill 
		\begin{subfigure}[t]{0.31\textwidth}
			\includegraphics[width=.95\linewidth,height=3.75cm]{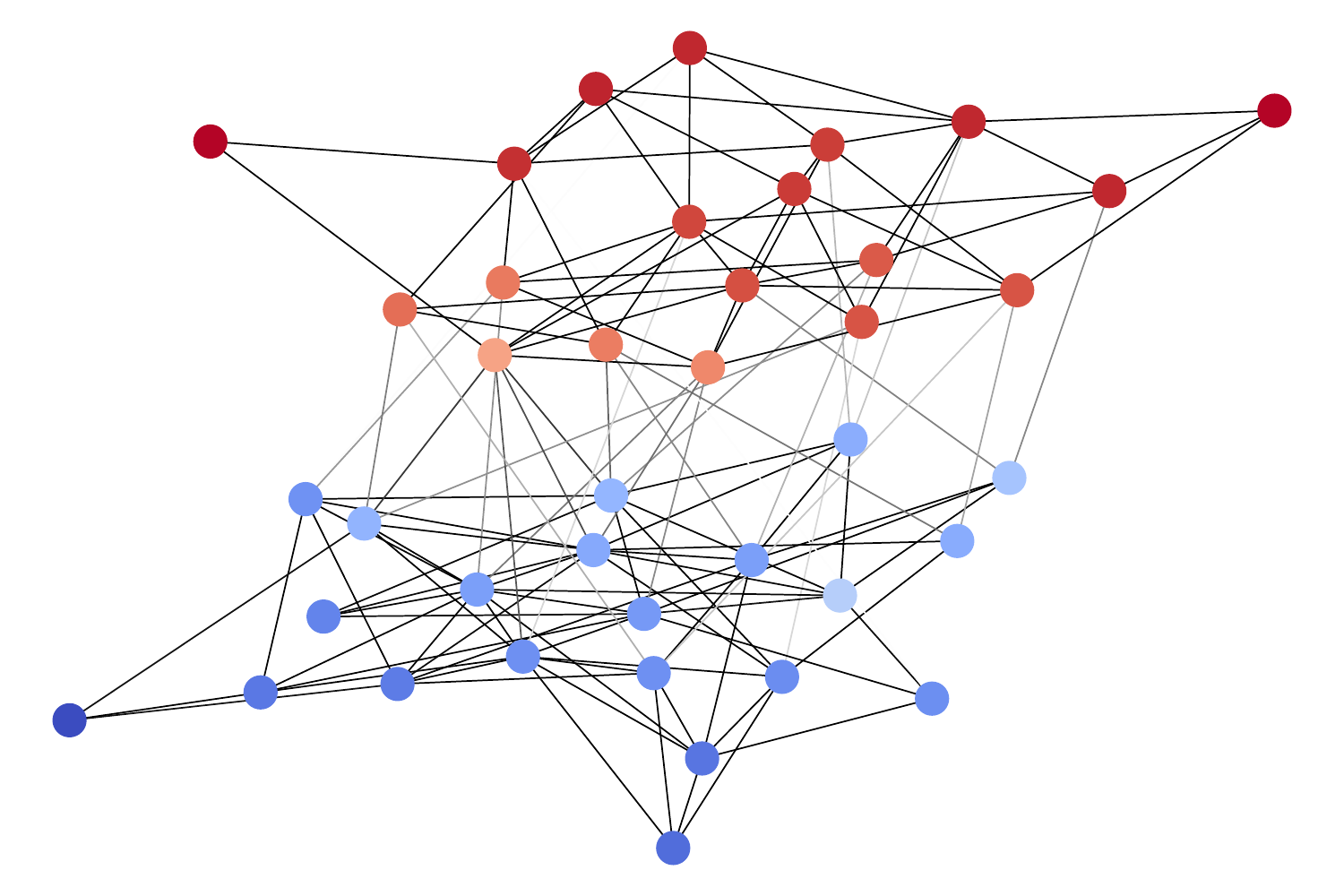}
			\caption{Graph after network administrator changes just 30\% of edge weight.} \label{fig:1c}
		\end{subfigure}
			\vspace{-.5em}
		\caption{Social network graphs after converging to equilibrium in the Friedkin-Johnsen opinion dynamics model. Node colors represent the opinions of individuals on an issue: dark red nodes have opinion close to $1$, while dark blue nodes have opinion close to $-1$. The weight of an edge (i.e, strength of connection between two individuals) is expressed by its shade.\\
			In the middle and right networks, we introduce a \emph{network administrator} who is allowed to make small changes to the network, and is incentivized to connect users with content that is similar to their opinion. After reweighting edges by just a small amount (i.e. filtering social content), the network administrator's actions increase a standard measure of opinion polarization in these graphs by 180\% and 260\%, respectively. This illustrates the formation of a ``filter bubble" in the network.}
		\label{fig:teaser}
		\vspace{-.5em}
	\end{figure}	
	
	\section{Introduction}
	
	The past decade has seen an explosion in social media use and importance. Online social networks, which enable users to instantly broadcast information about their daily lives and opinions to a large audience, are used by billions of people worldwide. Social media is also used to access news \cite{ShearerMatsa:2018}, review products and restaurants, find health and wellness recommendations \cite{SmithChristakis:2008}, and more. 
	
	Social networks, along with the world wide web in general, have made our world more connected. It has been widely established that social networks and online media increase the diversity of information and opinions that individuals are exposed to \cite{Brundidge:2010,Kim:2011,LeeChoiKim:2014}. In many ways, the widespread adoption of online social networks has resulted in significant positive progress towards fulfilling Facebook's mission of ``bringing the world closer together''.
	
	\subsection{The puzzle of polarization}
	Surprisingly, while they enable access to a diverse array of information, social networks have also been widely associated with \emph{increased polarization} in society across many issues \cite{FlaxmanGoelRao:2016}, including politics \cite{AdamicGlance:2005, ConoverRatkiewiczFrancisco:2011, Baer:2016}, science \cite{McCrightDunlap:2011}, and healthcare \cite{Holone:2016}.
	Somehow, despite the exposure to a wide variety of opinions and perspectives, individuals form polarized clusters, unable to reach consensus with one another. In politics, increased polarization has been blamed for legislative deadlock, erratic policies, and decreased trust and engagement in the democratic process  \cite{LaymanCarseyHorowitz:2006, Binder:2014}.
	
	
	There have been many efforts to understand this seemingly counterintuitive phenomenon of increased societal polarization. Classical psychological theory asserts that polarization arises from ``biased assimilation'' \cite{LordRossLepper:1979}, i.e. individuals are more likely to trust and share information that already aligns with their views. Isolated examples of intense polarization, such as the 2016 US presidential election and Brexit, can be partially explained by historical, cultural, and ideological factors \cite{HarePoole:2014}. However, when such examples are considered in bulk, it becomes clear that changes in social dynamics arising from the increased use of social media must constitute a major contributing factor to the phenomenon of polarization \cite{Baer:2016, Jackson:2017}. 
%
	
	
	
	\subsection{Filter bubbles}
	An influential idea put forward by Eli Pariser suggests an important mechanism for explaining why social media increases societal polarization \cite{Pariser:2011}. According to Pariser, preferential attention to viewpoints similar to those already held by an individual is \emph{explicitly encouraged} by social media companies: to increase metrics like engagement and ad revenue, recommendation systems tend to connect users with information already similar to their current beliefs. 
	
	Such recommendations can be direct: friend or follow suggestions on platforms like Facebook or Twitter. Or they can be more subtle: chronological ``news feeds'' on social media have universally been replaced with individually filtered and sorted  feeds which connect users with posts that they are most likely to engage with \cite{Evans:2018}.
	By recommending such content, social network companies create ``echo chambers" of similar-minded users. 
	Owing to their root cause -- the external filtering of content shown to a user -- Pariser called these echo chambers \emph{filter bubbles}. 
	
	The danger of filter bubbles was recently highlighted by Apple CEO Tim Cook in a commencement speech at Tulane University \cite{Eadicicco:2019}. Filter bubbles have been blamed for the spread of fake news during the Brexit referendum and the 2016 U.S. presidential election \cite{Jackson:2017}, protests against immigration in Europe \cite{GeschkeLorenzHoltz:2019}, and even measles outbreaks in 2014 and 2015 \cite{Holone:2016}. In each of these incidents, instead of bringing diverse groups of users together, social media has reinforced differences between groups and wedged them apart. 
	
	
	
	At least... that's the theory. While Pariser's ideas make logical sense, the magnitude of the ``filter bubble effect'' has been disputed or questioned for lack of evidence \cite{Boutin:2011,NguyenHuiHarper:2014,Barbera2014,VaccariValerianiBarberaJostNagler2016,Zuiderveen-BorgesiusTrillingMoeller:2016,HosanagarFlederLeeBuja2014}.
	
	\subsection{Our contributions}
	The goal of this paper is to better understand filter bubbles, and ultimately, to place Pariser's theory on firmer ground.
	We do so by developing a mathematical framework for studying the effect of filter bubbles on polarization in social networks, relying on well-established analytical models for \emph{opinion dynamics} \cite{DasGollapudiMunagala:2014}. 
	
	Such models provide simple rules that capture how opinions form and propagate in a social network. The network itself is typically modeled as a weighted graph: nodes are individuals and social connections are represented by edges, with higher weight for relationships with increased interaction.
	We work specifically with the well-studied Friedkin-Johnsen opinion dynamics model, which models an individual's opinion on an issue as a continuous value between $-1$ and $1$, and assumes that, as time progresses, individuals update their opinions based on the average opinion of their social connections \cite{FriedkinJohnsen:1990}. The Friedkin-Johnsen model has been used successfully to study polarization in social networks \cite{BindelKleinbergOren:2015,MuscoMuscoTsourakakis:2018,ChenLijffijtDe-Bie:2018,ChenLijffijtDe-Bie:2018a}. 
	
	Our contribution is to modify the model by adding an external force: a \textbf{network administrator} who filters social interaction between users. Based on modern recommendation systems \cite{Aggarwal2016}, the network administrator makes small changes to edge weights in the network, which correspond to slightly increasing or decreasing interaction between specific individuals (e.g.  by tuning a news feed algorithm). The administrator's goal is to connect users with content they likely agree with, and therefore increase user engagement. Formally, we model this goal by assuming that the network administrator seeks to minimize a standard measure of \emph{disagreement} in the social network. As individuals update their opinions according to the Friedkin-Johnsen dynamics, the administrator repeatedly adjusts the underlying network graph to achieve its own goal.
	
	Using our model, we establish a number of experimental and theoretical results which suggest that content filtering by a network administrator can significantly increase polarization, even when changes to the network are highly constrained (and perhaps unnoticeable by users). 
First, we apply our augmented opinion dynamics model to real-world social networks obtained from Twitter and Reddit. When the network administrator changes only $40\%$ of the total edge weight in the network, polarization increases by more than a factor of $40\times$. These results are striking---they suggest that social networks are very sensitive to influence by filtering. As illustrated in Figure \ref{fig:teaser}, even minor content filtering by the network administrator can create significant ``filter bubbles'', just as Pariser predicted \cite{Pariser:2011}.

	Next, to better understand the sensitivity of social networks to filtering, we study a standard generative model for social networks: the stochastic block model \cite{Abbe:2018}. We show that, with high probability, any network generated from the stochastic block model is in a state of \textbf{fragile consensus}: that is, under the Friedkin-Johnsen dynamics, the network will exhibit low polarization, but can become highly polarized after only a minor adjustment of edge weights. Our findings give theoretical justification for why a network administrator can greatly increase polarization in real-world networks.
	
	
	
Finally, ending on an optimistic note, we show experimentally that a simple modification to the incentives of the network administrator can greatly mitigate the filter bubble effect. Surprisingly, our proposed solution also minimally affects the incentives of the network administrator---its objective, user disagreement, is only increased by at most $5\%$.
	
	\subsection{Prior work}
	
	\textbf{Minimizing polarization in social networks.} There has been a substantial amount of recent work which uses opinion dynamics models to study polarization in social networks. \cite{MatakosTerziTsaparas2017} first defines polarization in the Friedkin-Johnsen model, and gives an algorithm for reducing polarization in social networks. \cite{MuscoMuscoTsourakakis:2018} and \cite{ChenLijffijtDe-Bie:2018} give methods for finding network structures which minimize different functions involving polarization and disagreement in the Friedkin-Johnsen model. Our work differs from this prior work in that we study network modifications which \emph{increase} polarization, rather than decreasing it. Moreover, we study how such modifications arise even when the network administrator is not explicitly incentivized to change polarization.
	
Other opinion dynamics models and metrics have also been used to study network polarization. \cite{AslayMatakosGalbrun:2018} gives an algorithm for mitigating filter bubbles in an influence maximization setting. \cite{GarimellaMoralesGionisMathioudakis2018} studies ``controversy" in the Friedkin-Johnsen model, a metric related to polarization,  and \cite{GarimellaGionisParotsidis:2017} gives an algorithm for reducing controversy in networks.
	
	\textbf{Modeling filter bubbles and recommendation systems.} Biased assimilation, which is when users gravitate towards viewpoints similar to their own, has  been argued as one cause of increased polarization in social networks.
	By generalizing the classic DeGroot model \cite{DeGroot:1974} of opinion formation, \cite{DandekarGoelLee:2013} provides theoretical support for the biased assimilation phenomenon and analyzes the interaction of three recommendation systems on biased assimilation.
	\cite{Del-VicarioScalaCaldarelli:2017} models biased assimilation in social networks using a variant of the Bounded Confidence Model (BCM) \cite{HegselmannKrause2002}, an opinion dynamics model that does not assume a latent graph structure between users.
	Most similar to our work, \cite{GeschkeLorenzHoltz:2019} creates a variant of the BCM that models biased assimilation, homophily, and algorithmic filtering, and shows how echo chambers can arise as a result of these factors. 
	\cite{ChaneyStewartEngelhardt2018} studies the more general problem of how recommendation systems increase homogeneity of user behavior.

	\subsection{Notation and Preliminaries}
	We use bold letters to denote vectors, e.g. $\bv{a}$. The $i^\text{th}$ entry of $\bv{a}$ is denoted $a_i$. For a matrix $A$, $A_{ij}$ is the entry in the $i^\text{th}$ row and $j^\text{th}$ column. For a vector $\bv{a} \in \R^n$, let $\diag(\bv{a})$ return an $n\times n$ diagonal matrix with the $i^\text{th}$ diagonal entry equal to $a_i$. For a matrix $A \in \R^{n\times d}$, let  $\rsum(A)$ return a vector whose $i^\text{th}$ entry is equal to the sum of all entries in $A$'s $i^\text{th}$ row. We use $I_{n\times n}$ to denote a dimension $n$ identity matrix, and $\bv{1}_n$ to denote the all ones column vector, with the subscript omitted when dimension is clear from context.
	
	Every real {symmetric} matrix $A \in \R^{n\times n}$ has an orthogonal eigendecomposition $A = U\Lambda U^T$ where $U \in \R^{n\times n}$ is orthonormal (i.e $U^TU = UU^T = I$) and $\Lambda$ is diagonal, with real valued entries $\lambda_1 \leq \lambda_2 \leq \ldots \leq \lambda_n$ equal to $A$'s eigenvalues.
	We say a symmetric matrix is positive semidefinite (PSD) is all of its eigenvalues are non-negative (i.e. $\lambda_1 \geq 0$). We use $\preceq$ to denote the standard Loewner ordering: $M\preceq N$ indicates that $N - M$ is PSD. For a square matrix $M$, $\|M\|_2$ denotes the spectral norm of $M$ and $\|M\|_F$ denotes the Frobenius norm. For a vector $v$, $||v ||_2$ denotes the $L^2$ norm.
	
	\subsection{Road Map}
	\begin{description}[style=unboxed,leftmargin=0cm,itemsep=0mm]
		\item[Section \ref{sec:fj}] Introduce preliminaries on Freidkin-Johnsen opinion dynamics, which form a basis for modeling  filter bubbles.
		\item[Section \ref{sec:filter_bubbles}] Introduce our central ``network administrator dynamics'' and establish experimentally that content filtering can significantly increase polarization in social networks. 
		\item[Section \ref{sec:sbm}] Explore these findings theoretically by showing that stochastic block model graphs exhibit a ``fragile consensus'' which is easily disrupted by outside influence.
			\item[Section \ref{sec:remedy}] Discuss a small modification to the content filtering process that can mitigate the effect of filter bubbles while still being beneficial for the network administrator.
		\item[Section \ref{sec:open}] Briefly discuss future directions of study.
	\end{description}

	\section{Modeling Opinion Formation}
	\label{sec:fj}
	One productive approach towards understanding the dynamics of consensus and polarization in social networks has been to develop simple mathematical models to explain how information and ideas spread in these networks. 
	
	While there are a variety of models in the literature, we use the Friedkin-Johnsen opinion dynamics model, which has been used to study polarization in recent work \cite{MatakosTerziTsaparas2017,MuscoMuscoTsourakakis:2018,ChenLijffijtDe-Bie:2018}.
	
	\subsection{Friedkin-Johnsen Dynamics}
	Concretely, the Friedkin-Johnsen (FJ) dynamics applies to any social network that can be modeled as an undirected, weighted graph $G$. Let $\{v_1, \ldots, v_n\}$ denote $G$'s nodes and for all $i \neq j$, let $w_{ij} \geq 0$ denote the weight of undirected edge $(i,j)$ between nodes $v_i$ and $v_j$. 
	Let $d_i = \sum_{j\neq i} w_{ij}$ be the degree of node $v_i$. 
	
	The FJ dynamics model the propagation of an opinion on an issue during a discrete set of time steps $t\in 0, 1, \ldots, T$. The issue may be specific (Do you believe that humans contribute to climate change?) or it may encode a broad ideology (Do your political views align most with conservative or liberal politicians in the US?).
	
	In either case, the FJ dynamics assume that the issue has exactly two poles, with an individual's opinion encoded by a continuous real value in $[-1,1]$. $-1$ and $1$ represent the most extreme opinions in either direction, while $0$ represents a neutral opinion.
	Each node $v_i$ holds an ``innate'' (or internal) opinion $s_i \in [-1,1]$ on the issue. The internal opinion vector $\bv{s}=[s_1, \ldots, s_n]$ does not change over time. It can be viewed as the opinion an individual would hold in a social vacuum, with no outside influence from others. The value of $s_i$ might depend on the background, geographic location, religion, race, or other circumstances about individual $i$. 
	
	In addition to an innate opinion, for every time $t$, each node is associated with an ``expressed'' or ``current'' opinion $z_i^{(t)} \in [-1,1]$, which changes over time. Specifically, the FJ dynamics evolves according to the update rule:
	\begin{align}
	\label{eq:fj_dynamics}
	z_i^{(t)} = \frac{s_i + \sum_{j \neq i} w_{ij} z_j^{(t-1)}}{d_i + 1}.
	\end{align}
That is, at each time step, each node adopts a new expressed opinion which is the average of its own innate opinion and the opinion of its neighbors. For a given graph $G$ and innate opinion vector $\bv{s}$, it is well known that the FJ dynamics converges to an equilibrium set of opinions \cite{BindelKleinbergOren:2015}, which we denote 
	\begin{align*}
	\bv{z}^* = \lim_{t\rightarrow \infty} \bv{z}^{(t)}.
	\end{align*}
	
	It will be helpful to express the FJ dynamics in a linear algebraic way. Let $A\in \R^{n\times n}$ be the adjacency matrix of $G$, with $A_{ij} = A_{ji} = w_{ij}$ and let $D$ be a diagonal matrix with $D_{ii} = d_i$. Let $L = D -A$ be the graph Laplacian of $G$. Then we can see that \eqref{eq:fj_dynamics} is equivalent to
	\begin{align}
	\label{eq:fj_dynamics_matrix}
	\bv{z}^{(t)} = (D + I)^{-1}(A\bv{z}^{(t-1)} + \bv{s}),
	\end{align}
	where we denote $\bv{z}^{(t)} =[z_1^{(t)}, \ldots, z_n^{(t)}]$.
	From this expression, it is not hard to check that
	\begin{align}
	\label{eq:fj_solution}
	\bv{z}^{*} = (L+I)^{-1} \bv{s}.
	\end{align}
	
	\medskip
	\noindent\textbf{Alternative Models.}
	The Friedkin-Johnsen opinion dynamics model is a variation of DeGroot's classical model for consensus formation in social network \cite{DeGroot:1974}. The distinguishing characteristic of the FJ model is the addition of the \emph{innate opinions} encoded in $\bv{s}$. Unlike the DeGroot model, which always converges in a consensus when $G$ is connected (i.e., $z^*_i = z^*_j$ for all $i,j$) , innate opinions allow for a richer set of equilibrium opinions. In particular, $\bv{z}^*$ will typically contain opinions ranging continuously between $-1$ and $1$. 
	
	Compared to DeGroot, the FJ dynamics more accurately model a world where an individual's opinion (e.g. on a political issue) is not shaped solely by social influence, but also by an individual's particular background, beliefs, or life circumstances. FJ dynamics are often studied in economics and game theory as an example of a game with price of anarchy greater than one \cite{BindelKleinbergOren:2015}. 
	Other variations on the model include additional variables\cite{HegselmannKrauseothers:2002}, for example, allowing the ``stubbornness'' of an individual to vary \cite{AbebeKleinbergParkes:2018,ChenLijffijtDe-Bie:2018a}, or adding additional terms to Equation \eqref{eq:fj_dynamics} that indicate when an individual cares about the average network opinion as well as their neighbors' opinions \cite{EpitropouFotakisHoeferSkoulakis:2017}.

	
	There also exist many models for opinion formation that fall outside of DeGroot's original framework. Several models involve \emph{discrete} instead of continuously valued opinions. We refer to reader to the overview and discussion of different proposals in \cite{DasGollapudiMunagala:2014}. In this paper, we focus on the original FJ dynamics, which are already rich enough to provide several interesting insights on the dynamics of polarization, filter bubbles, and echo chambers.

	\subsection{Polarization, Disagreement, and Internal Conflict}
	The fact that $\bv{z}^*$ does not always contain a single consensus opinion makes the FJ model suited to understanding how polarization arises on specific issues. Formally, we define polarization as the variance of a given set of opinions.
	\begin{definition}[Polarization, $\pP_{\bv{z}}$]
		\label{def:polarization}
		For a vector of $n$ opinions $\bv{z} \in [-1,1]^n$, let $\mean(\bv{z}) = \frac{1}{n}\sum_{j=1}^n z_j$ be the mean opinion in $\bv{z}$. 
		\begin{align*}
		\pP_{\bv{z}} \eqdef \sum_{i=1}^n (z_i - \mean(\bv{z}))^2.
		\end{align*}
	\end{definition}
	$\pP_{\bv{z}}$ ranges between $0$ when all opinions are equal and $n$ when half of the opinions in $\bv{z}$ equal $1$ and half equal $-1$. $\pP_{\bv{z}}$ was first proposed as a measure of polarization in \cite{MatakosTerziTsaparas2017}, and has since been used in other recent work studying polarization in FJ dynamics \cite{MuscoMuscoTsourakakis:2018,ChenLijffijtDe-Bie:2018}. While we focus on Definition \ref{def:polarization}, we refer the interested reader to \cite{GarimellaMoralesGionisMathioudakis2018} for discussion of alternative measures of polarization.
	
	Under the FJ model, the polarization of the equilibrium set of opinions has a simple closed form. In particular, let $\bv{\overline{s}} = \bv{s} - \bv{1}\cdot \mean(\bv{s})$ be the mean centered set of innate opinions on a topic, and define $\bv{\overline{z}}$ similarly. Using that $\bv{1}$ is in the null-space of any graph Laplacian $L$, it is easy to check (see \cite{MuscoMuscoTsourakakis:2018} for details) that $\mean(\bv{z}) = \mean(\bv{s})$
	and thus $\bv{\overline{z}}^* = (L+I)^{-1} \bv{\overline{s}}.$
	It follows that:
	\begin{align}\label{eq:equilibrium_polarization}
	\pP_{\bv{z}^*} = \bv{\overline{s}}^T (L+I)^{-2}\bv{\overline{s}}.
	\end{align}
	
	In addition to polarization, we define two other quantities of interest involving opinions in a social network. Both have appeared repeatedly in studies involving the Friedkin-Johnsen dynamics \cite{AbebeKleinbergParkes:2018,MuscoMuscoTsourakakis:2018,ChenLijffijtDe-Bie:2018}. 
	
	The first quantity measures how much node $i$'s opinion differs from those of its neighbors.
	\begin{definition}[Local Disagreement, $\dD_{G,\bv{z},i}$]
		\label{def:local_disagreement}
		For $i \in 1,\ldots,n$, a vector of opinions $\bv{z} \in [-1,1]^n$, and social network graph $G$,
		\begin{align*}
		\dD_{G,\bv{z},i} \eqdef \sum_{j \in 1, \ldots n, j \neq i} w_{ij} (z_i - z_j)^2.
		\end{align*}
	\end{definition}
	We also define an aggregate measure of disagreement.
	\begin{definition}[Global Disagreement, $\dD_{G,\bv{z}}$]
		\label{def:disagreement}
		For a vector of opinions $\bv{z} \in [-1,1]^n$, and social network graph $G$,
		\begin{align*}
		\dD_{G,\bv{z}} \eqdef \frac{1}{2} \cdot \sum_{i=1}^n \dD_{G,\bv{z},i}.
		\end{align*}
		The factor of $1/2$ is included so that each edge $(i,j)$ is only counted once. When $G$ has graph Laplacian $L$, it can be checked (see e.g. \cite{MuscoMuscoTsourakakis:2018}) that $\dD_{G,\bv{z}} = \bv{z}^T L \bv{z} =  \bv{\overline{z}}^T L \bv{\overline{z}}$.
	\end{definition}
	Disagreement measures how misaligned each node's opinion is with the opinions of its neighbors. We are also interested in how misaligned a node's expressed opinion is with its innate opinion.
	\begin{definition}[Local Internal Conflict, $\iI_{\bv{z},\bv{s}, i}$]
		\label{def:local_ic}
		For $i \in 1,\ldots,n$, a vector of expressed opinions $\bv{z} \in [-1,1]^n$, and a vector of innate opinions $\bv{s} \in [-1,1]^n$, 
		\begin{align*}
		I_{\bv{z},\bv{s},i} \eqdef (z_i - s_i)^2.
		\end{align*}
	\end{definition}
	We also define an aggregate measure of internal conflict.
	\begin{definition}[Global Internal Conflict, $I_{\bv{z},\bv{s}}$]
		\label{def:ic}
		For a vector of expressed opinions $\bv{z} \in [-1,1]^n$, and a vector of innate opinions $\bv{s} \in [-1,1]^n$, 
		\begin{align*}
		\iI_{\bv{z},\bv{s}} \eqdef \sum_{i=1}^n \iI_{\bv{z},\bv{s},i}  = \|\bv{z} - \bv{s}\|_2^2.
		\end{align*}
		Since $\mean(\bv{z}) = \mean(\bv{s})$, we equivalently have $\iI_{\bv{z},\bv{s}}  =  \|\bv{\overline{z}} - \bv{\overline{s}}\|_2^2.$
	\end{definition}
	
	We can rewrite both the Friedkin-Johnsen update rule and equilibrium opinion vector as solutions to optimization problems involving minimizing disagreement and internal conflict.
	
	\begin{claim}
		The Friedkin-Johnsen dynamics update rule (Equation \ref{eq:fj_dynamics}) is equivalent to 
		\begin{equation}
		\label{eq:fj_update_rule_opt}
		z_i^{(t)} = \argmin_z \dD_{G,\bv{z},i}+\iI_{\bv{z},\bv{s},i}.
		\end{equation}
		The equilibrium opinion vector $\bv{z}^*$ (Equation \ref{eq:fj_solution}) is equivalent to
		\begin{equation}
		\label{eq:fj_equilibrium_opt}
		\bv{z}^* = \argmin_z \dD_{G,\bv{z}}+\iI_{\bv{z},\bv{s}}.
		\end{equation}
	\end{claim}
	
	It was also observed in \cite{ChenLijffijtDe-Bie:2018} that polarization, disagreement, and internal conflict obey a ``conservation law'' in the Friedkin-Johnsen dynamics.
	\begin{claim}[Conservation law]
		For any network graph $G$ with Laplacian $L$, innate opinions $\bv{s} \in [-1,1]^n$, and equilibrium opinions $\bv{z}^* = (L+I)^{-1}\bv{s}$, 
		\begin{align}
		\label{eq:fj_conservation_law}
		\pP_{\bv{z}^*} + 2\cdot \dD_{G,\bv{z}^*} + \mathcal{I}_{\bv{z}^*,\bv{s}} = \bv{\overline{s}}^T  \bv{\overline{s}}.
		\end{align}
	\end{claim}
	
Now, combining Equations \eqref{eq:fj_equilibrium_opt} and \eqref{eq:fj_conservation_law} tells us that $\bv{z}^*$, the equilibrium solution of the Friedkin-Johnsen dynamics, maximizes polarization plus disagreement.
	\begin{equation}
	\label{eq:fj_equilibrium_opt_max}
	\bv{z}^* = \argmax_z \pP_z + \dD_{G,\bv{z}}.
	\end{equation}
	
	Now suppose we add another actor, whose goal is to minimize disagreement, to the model. Informally, since the users of the network are maximizing polarization + disagreement, and this other actor is minimizing disagreement, one would expect polarization to increase. This intuitive observation motivates the network administrator dynamics, described below, as a vehicle for the emergence of filter bubbles in a network.
	

	\section{The Emergence of Filter Bubbles}
	\label{sec:filter_bubbles}
	
	We introduce another actor to the Friedkin-Johnsen opinion dynamics, the \textbf{network administrator}. The network administrator increases user engagement via personalized filtering, or showing users content that they are more likely to agree with. In the Friedkin-Johnsen model, this corresponds to the network administrator reducing disagreement by making changes to the edge weights of the graph (e.g. users see more content from users with similar opinions, and less content from users with very different opinions). 
	
	

	\subsection{Network Administrator Dynamics}
	
	Formally, our extension of the Friedkin-Johnsen dynamics has two actors: users, who change their expressed opinions $\bv{z}$, and a network administrator, who changes the graph $G$. The \emph{network administrator dynamics} are as follows.
	
	
	\begin{tcolorbox}[pad at break=1mm] 
		\textbf{Network Administrator Dynamics.} \\
		Given initial graph $G^{(0)} = G$ and initial opinions $z^{(0)} = s$, in each round $r=1, 2, 3, \dots$
		\begin{itemize}
			\item
			First, the users adopt new expressed opinions $z^{(r)}$. These opinions are the equilibrium opinions (Equation \ref{eq:fj_solution}) of the FJ dynamics model applied to $G^{(r-1)}$:
			\begin{equation}
			\label{eq:na_dyn_z}
			z^{(r)} = (L^{(r-1)}+I)^{-1}s.
			\end{equation}
			Here $L^{(r-1)}$ is the Laplacian of $G^{(r-1)}$.
			\item
			Then, given user opinions $z^{(r)}$, the network administrator minimizes disagreement by modifying the graph, subject to certain restrictions:
			\begin{equation} \label{eq:na_game_update_w}
			G^{(r)} = \argmin_{G \in S}\; \dD_{G, z^{(r)}}.
			\end{equation}
			$S$ is the constrained set of graphs the network administrator is allowed to change to.
		\end{itemize}
		
	\end{tcolorbox}
	
	\subsubsection{Restricting changes to the graph}
	
	$S$, the set of all graphs the network admin can modify the graph to (Equation \ref{eq:na_game_update_w}), should reflect realistic changes that a recommender system would make. For example, if the network admin is unconstrained, then the network admin will simply set $w_{ij} = 0$ for all edges $(i,j)$, as the empty graph minimizes disagreement. This is entirely unrealistic, however, as a social network would never eliminate all connections between users.
	In our experiments, we define $S$ as follows:
	
	\begin{tcolorbox}[breakable,pad at break=1mm] 
		\textbf{Constraints on the network administrator.} \\
		Given $\epsilon > 0$ and initial graph $\overline{G}$ with adjacency matrix $\overline{W}$, let $S$ contain all graphs with adjacency matrix $W$ satisfying:
		
		\begin{enumerate}
			\item
			$||W-\overline{W}||_F < \epsilon \cdot ||\overline{W}||_F$.
			
			\item
			$\sum_j W_{ij} = \sum_j (\overline{W})_{ij}$ for all $i$, i.e. the degree of each vertex should not change.

			%
		\end{enumerate}
	\end{tcolorbox}

			The first constraint prevents the network administrator from making large changes to the initial graph $\overline{W}$. Here, $\epsilon$ represents an $L^2$ constraint parameter for how much the network administrator can change edge weight in the network. 
			The second constraint restricts the network administrator to only making changes that maintain the total level of interaction for every user. Otherwise, the network administrator would reduce disagreement by decreasing the total amount of edge weight in the graph---corresponding to having people spend less time on the network---which is not realistic.
	
	
	Note that, since $S$ gives a convex set over adjacency matrices and $\dD_{G, z^{(r)}}$ is a convex function (as a function of the adjacency matrix of $G$), the minimization problem in Equation \eqref{eq:na_game_update_w} has a unique solution, eliminating any ambiguity for the network administrator.

	\subsubsection{Convergence}
	Although it is not immediately obvious, the Network Administrator Dynamics do converge. In each round, the users are minimizing disagreement + internal conflict (Equation \ref{eq:fj_equilibrium_opt}), while the network admin is minimizing disagreement (Equation \ref{eq:na_game_update_w}). Thus, we can view the Network Administrator Dynamics as alternating minimization on disagreement + internal conflict:
	\begin{equation}
	\argmin_{z\in\mathbb{R}^n,W \in S} \dD_{G,z} + \iI_{z,s}.
	\end{equation}
	
	While $\dD_{G,z} + \iI_{z,s}$ is not convex in both $z$ and $W$, it is convex in one variable when the other is fixed. Because our constraints on $W$ are also convex, alternating minimization will converge to a stationary point of $\dD_{G,z} + \iI_{z,s}$ \cite{Bertsakas1999,Beck2015}. Moreover, while the convergence point is not guaranteed to be the global minima of $\dD_{G,z} + \iI_{z,s}$, we empirically find that alternating minimization converges to a better solution than well-known optimization methods such as sequential quadratic programming \cite{BoggsTolle1995} and DMCP \cite{ShenDiamondUdellGuBoyd2016}.
	
	\subsection{Experiments}
	
	\begin{figure*}[h]
		\begin{subfigure}{0.5\textwidth}
			\centering
			\includegraphics[height=4.5cm]{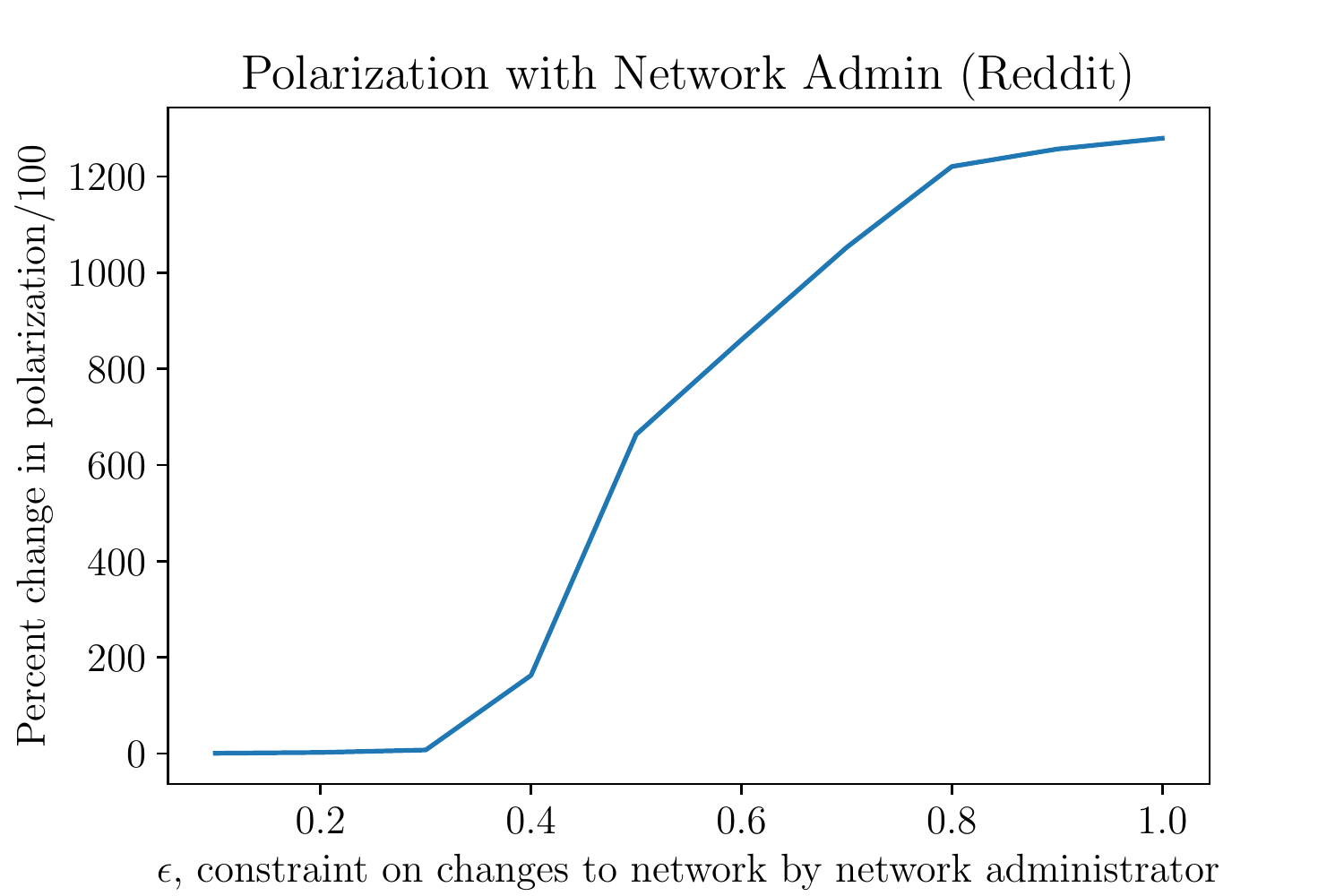}
			\caption{Change in polarization, Reddit network}  \label{subfig:reddit_pol}
		\end{subfigure}
		\begin{subfigure}{0.5\textwidth}
			\centering
			\includegraphics[height=4.5cm]{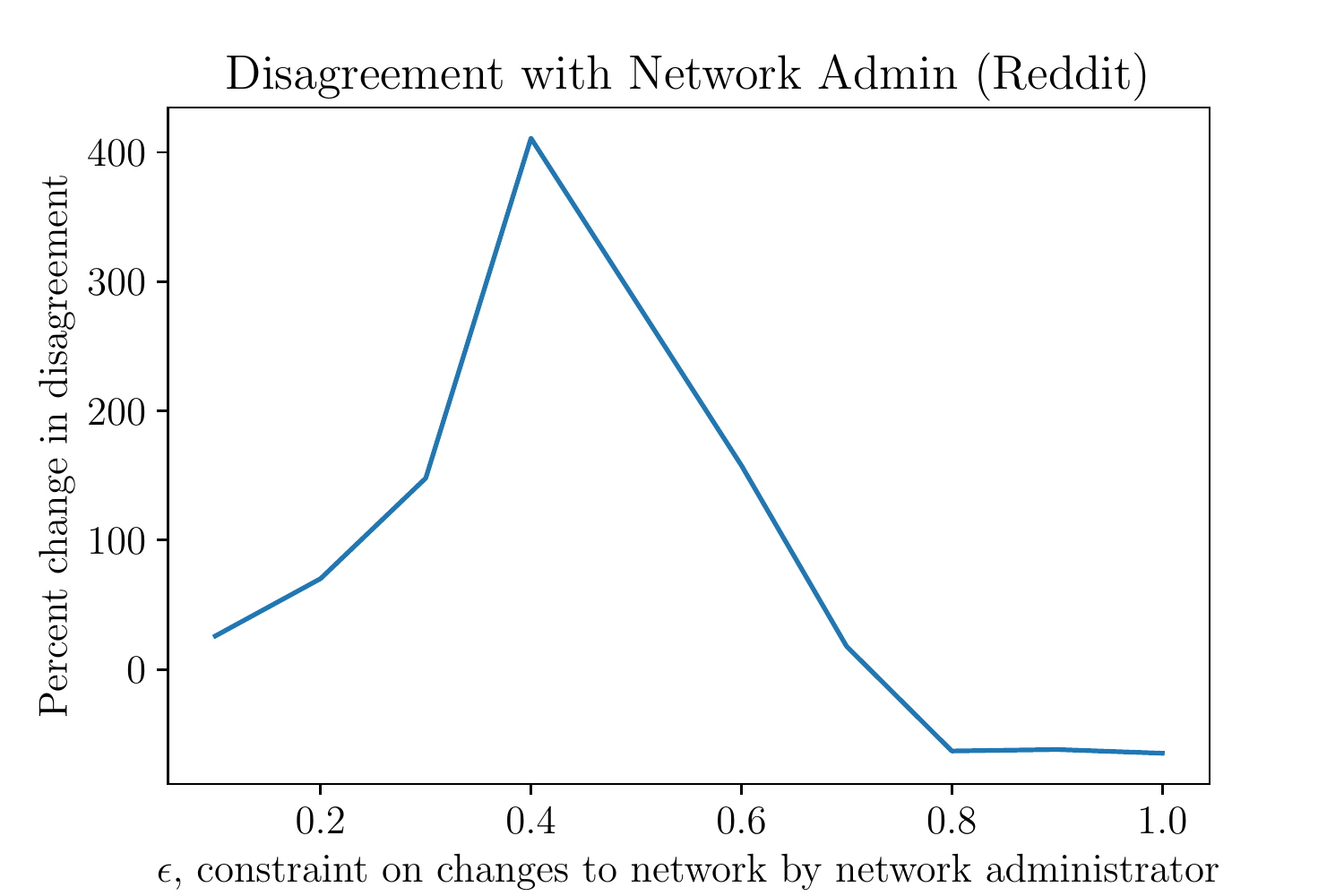}
			\caption{Change in disagreement, Reddit network} \label{subfig:reddit_disagg}
		\end{subfigure}
		\begin{subfigure}{0.5\textwidth}
			\centering
			\includegraphics[height=4.5cm]{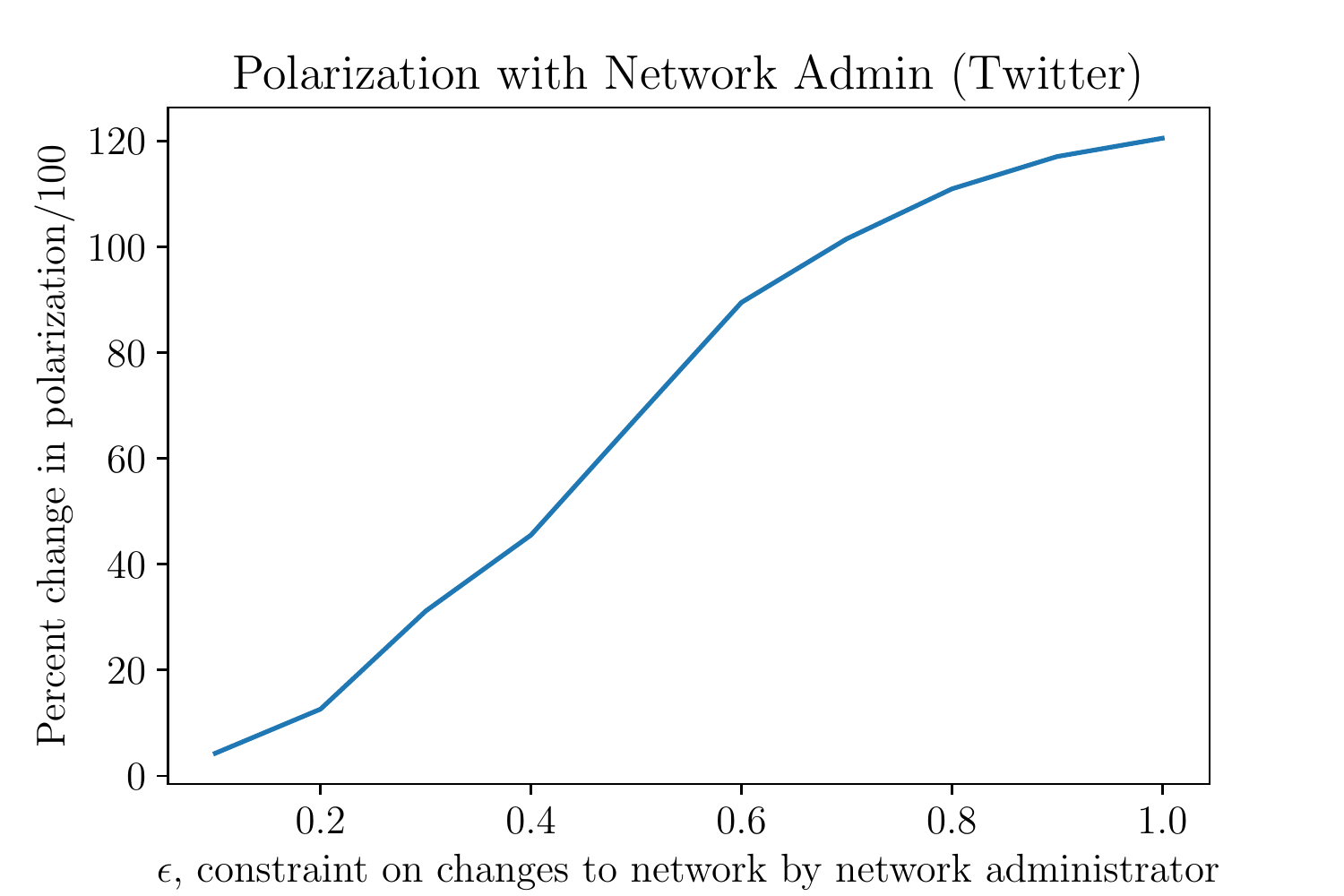}
			\caption{Change in polarization, Twitter network} \label{subfig:twitter_pol}
		\end{subfigure}
		\begin{subfigure}{0.5\textwidth}
			\centering
			\includegraphics[height=4.5cm]{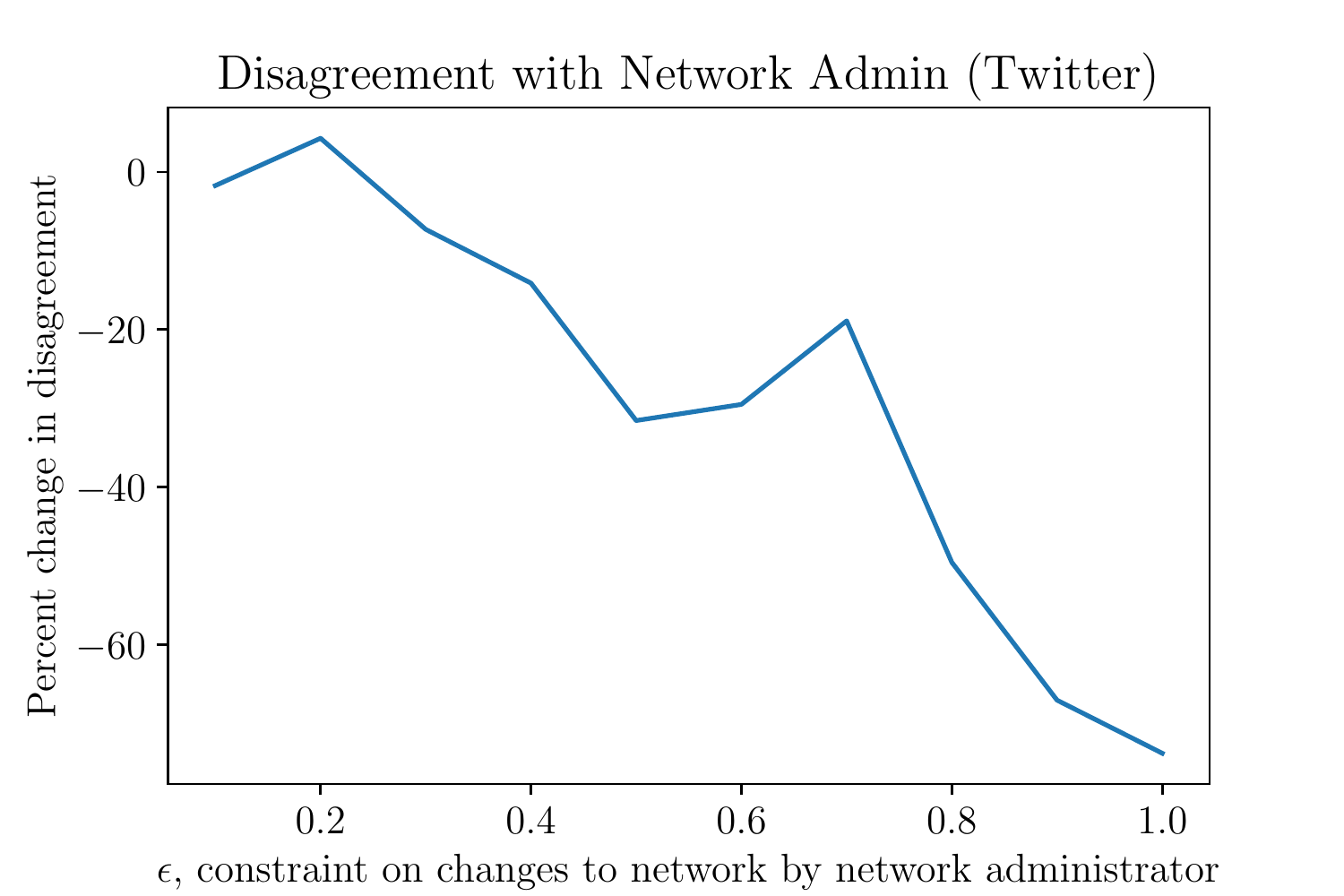}
			\caption{Change in disagreement, Twitter network} \label{subfig:twitter_disagg}
		\end{subfigure}
		\caption{Applying network administrator dynamics to real-world social networks. Details in Section \ref{sec:filter_bubbles}.} \label{fig:reddit_twitter}
	\end{figure*}
	
	Using two real-world networks, we show that content filtering by the network administrator greatly increases polarization.
	
	\textbf{Datasets.} We use two real-world networks collected in \cite{De2014}, which were previously used to study polarization in \cite{MuscoMuscoTsourakakis:2018}. We briefly describe the datasets. More details can be found in \cite{De2014,MuscoMuscoTsourakakis:2018}. 
	
	\emph{Twitter} is a network with $n=548$ nodes and $m=3638$ edges. Edges correspond to user interactions. The network depicts the debate over the Delhi legislative assembly elections of 2013. 
	
	\emph{Reddit} is a network with $n = 556$ nodes and $m = 8969$ edges. Nodes are users who posted in the politics subreddit, and there is an edge between two users if there exist two subreddits (other than politics) that both users posted in during the given time period.
	
	In both networks, each user has multiple opinions associated to them, obtained via sentiment analysis on multiple posts. Similar to \cite{MuscoMuscoTsourakakis:2018}, we average each of these opinions to obtain an equilibrium expressed opinion $z^*_i$ for each user $i$. Inverting Equation \eqref{eq:fj_solution} yields innate opinions $\bv{s}=(L+I)\bv{z}$, which we clamp to $[-1,1]$. This yields a rough estimate of the innate opinions of each user, and provides a starting point for analyzing the dynamics of polarization.

	%
	%
	
	\textbf{Results.} Figure \ref{fig:reddit_twitter} shows our results applying the network administrator dynamics to the Reddit and Twitter datasets. For both networks, we calculate the increase in polarization after introducing the network administrator dynamics, relative to the polarization of the equilibrium opinions without the network administrator. We plot this polarization increase versus $\epsilon$, the $L^2$ parameter that specifies how much the network administrator can change the network. We also plot the increase in disagreement versus $\epsilon$.
	
Once $\epsilon$ is large enough, polarization rises greatly in both networks. For example, when $\epsilon = 0.5$, polarization increases by a factor of around $\bv{700\times}$ in the Reddit network, and a factor of around $\bv{60\times}$ in the Twitter network. While polarization increases in both networks, it is interesting to observe that the Twitter network is more resilient than the Reddit network. Surprisingly, for $\epsilon < 0.7$, disagreement also increases in the Reddit network---so the network administrator does not even accomplish its goal of reducing disagreement.


Overall, our experiments illustrate how recommender systems can greatly increase opinion polarization in social networks, and give experimental credence to the theory of filter bubbles \cite{Pariser:2011}.
	
	\section{Fragile Consensus in Social Network Graphs}
	\label{sec:sbm}
	
Our results in Section \ref{sec:filter_bubbles} establish that polarization in Friedkin-Johnsen opinion models can significantly increase even when the network administrator adjusts just a small amount of edge weight.
	
To better understand this empirical finding, we present a theoretical analysis of the \emph{sensitivity} of social networks  to outside influence. In this work we are most interested in the effect of ``filtering'' by a network administrator, but our analysis can also be applied to potential influence from advertisers \cite{KempeKleinbergTardos:2003,GionisTerziTsaparas:2013} or propaganda \cite{CarlettiFanelliGrolli:2006}.
We want to understand how easily such outside influence can affect the polarization of a network.
	
	\subsection{The Stochastic Block Model}
We consider a common generative model for networks that can lead to polarization: the stochastic block model (SBM) \cite{HollandLaskeyLeinhardt:1983}.
	
	\begin{definition}[Stochastic Block Model (SBM)]
		\label{def:sbm}
		The stochastic block model is a random graph model parametrized by $n$, the size of the communities, and $p, q$, the edge probabilities. The model generates a graph $G$ with $2n$ vertices, where the vertex set of $G$, $V = \{v_1, \ldots, v_{2n}\}$, is partitioned into two sets or ``communities'', $S = \{v_1, \ldots, v_n\}$ and $T = \{v_{n+1}, \ldots, v_{2n}\}$. Edges are generated as follows. For all $v_i, v_j \in V$:
		\begin{itemize}
			\item If $v_i, v_j \in S$ or $v_i, v_j \in T$, set $w_{ij} = 1$ with probability $p$, and $w_{ij} = 0$ otherwise.
			\item If $v_i \in S, v_j \in T$ or $v_i \in T, v_j \in S$, set $w_{ij} = 1$ with probability $q$, and $w_{ij} = 0$ otherwise.
		\end{itemize}
	\end{definition}
Also known as "planted partition model", the stochastic block model has as long history of study in statistics, machine learning, theoretical computer science, statistical physics, and a number of other areas. It has been used to study social dynamics, suggesting it as a natural choice for analyzing the dynamics of polarization \cite{BecchettiClementiNatale:2017, Mallmann-TrennMuscoMusco:2018}.
		We refer the reader to the survey in \cite{Abbe:2018} for a complete discussion of applications and prior theoretical work on the model. 
	
		There are many possible variations on Definition \ref{def:sbm}. For example, $S$ and $T$ may differ in size or $V$ may be partitioned into more than two communities. Our specific setup is both simple and well-suited to studying the dynamics of opinions with two poles, as in the Friedkin-Johnsen model.
	
	\subsection{Opinion Dynamics in the SBM}
	
	As in most work on the SBM, we consider the natural setting where $q < p$, i.e. the 
	probability of two nodes being connected is higher when the nodes are in the same community, and lower when they are in different communities. This setting results in a graph $G$ which is ``partitioned'': $G$ looks like two identically
	distributed Erd\H{o}s-R\'{e}nyi random graphs, connected by a small number of random edges.
	
We assume the nodes in $S$ have innate 
	opinions clustered near $-1$ (one end of the opinion spectrum), and the nodes in $T$ have innate opinions clustered near $1$, so that nodes with similar innate opinions are more likely to be connected by edges. 
This property, known as "homophily", is commonly observed in real-world social networks \cite{DandekarGoelLee:2013}. Homophily arises because innate opinions are often correlated with demographics like age, geographic location, and education level---demographics that influence the probability of two nodes being connected. 
	
	With the SBM chosen as a model for graphs which resemble real-world social networks, this section's main question is:
	
	\begin{quote}
		\textit{How sensitive is the equilibrium polarization of a Friedkin-Johnsen opinion dynamics to changes in the underlying social network graph $G$, when $G$ is generated from a SBM?}
	\end{quote}
	
	To answer this question, we analyze how the equilibrium  polarization of SBM networks depends on parameters $p$ and $q$. We show that polarization of the equilibrium opinions decreases \emph{quadratically} with $q$, which means that even networks with very few edges between $S$ and $T$ have low polarization. 
	
	Formally, let $A \in \R^{2n\times 2n}$, $D =  \diag(\rsum(A))$, and $L = D - A$, be the adjacency matrix, diagonal degree matrix, and Laplacian, respectively, of a graph $G$ drawn from the stochastic block model. For simplicity, assume the FJ dynamics with $\bv{s}$ set to completely polarized opinions, which perfectly correlate with a node $v_i$'s membership in either $S = \{v_1, \ldots v_n\}$ or $T = \{v_{n+1} \ldots v_{2n} \}$:
	\begin{align}
	\label{eq:our_s}
	s_i = 
	\begin{cases}
	1 & \text{ for } i \in 1, \ldots, n \\
	-1 & \text{ for } i \in n+1, \ldots, 2n
	\end{cases}
	\end{align} 

Our main result is below.
	\begin{theorem}[Fragile consensus in SBM networks]
		\label{thm:main_sbm} 
		Let $G$ be a graph generated by the SBM with $1/n \leq q \leq p$ and $p > c \log^4 n / n$ for some universal constant $c$.
		Let $\bv{s}$ be the innate opinion vector defined in Equation \eqref{eq:our_s}, and let $\bv{v}^*$ be the equilibrium opinion vector according to the FJ dynamics. Then for sufficiently large $n$,
		\begin{align*}
		C\frac{2n}{(2nq + 1)^2} \leq \mathcal{P}_{\bv{v}^*} \leq C'\frac{2n}{(2nq + 1)^2}
		\end{align*}
		with probability $97/100$, for universal constants $C, C'$ .
	\end{theorem}

Note that our assumptions on $q$ and $p$ are mild -- we simply need that, in expectation, each node has at least one connection outside of its home community, and $O(\log^4 n)$ connections within its home community. In real-world social networks, the average number of connections typically exceeds these minimum requirements.

\begin{figure}[h]
	\centering
	\includegraphics[width=.5\textwidth]{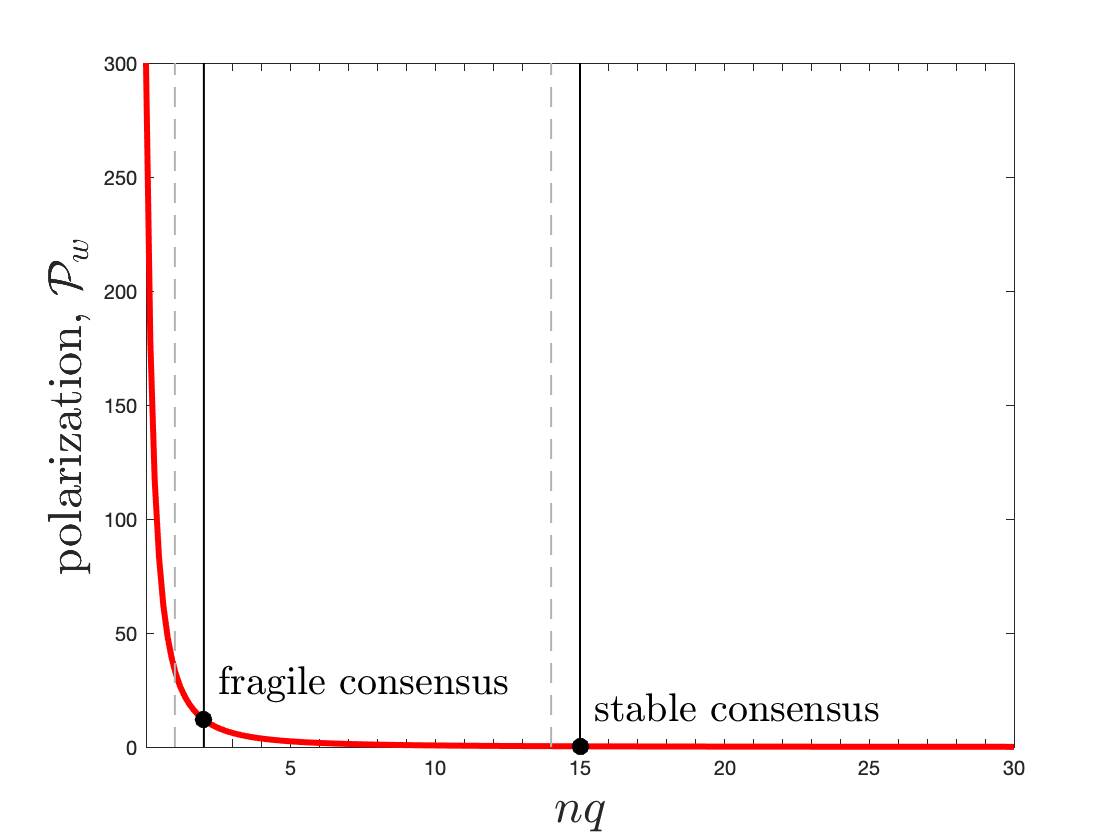}
\caption{The equilibirum polarization of a SBM social network plotted as a function of $nq$, i.e. the average number of ``out-of-group'' edges in the network per node. Polarization falls rapidly with $nq$, leading to a state of potentially fragile consensus, where removing a small number of edges from a network can vastly increase polarization.} 
	\label{fig:fragile_concensus}
\end{figure}

\medskip\noindent\textbf{Remarks.} Theorem \ref{thm:main_sbm} leads to two important observations. First, with high probability, the equilibrium polarization of a SBM network is \emph{independent} of $p$, the probability of generating an ``in-group" edge. This is highly counterintuitive: one would expect that increasing $p$ would decrease polarization, as each node would be surrounded by a larger proportion of like-minded nodes.
	
Second, when $nq$ is sufficiently large, polarization scales as \; $\sim \frac{2n}{(2nq)^2}$.
	Since the maximum polarization in a network with $2n$ nodes is $2n$, this says that the polarization of an SBM graph drops quadratically with $nq$, the expected number of ``out-of-group'' edges per node.
	This behavior is visualized in Figure \ref{fig:fragile_concensus}.
	
The second observation suggests an interesting view about social networks that are relatively un-polarized (i.e., are near consensus). In particular, it is possible for such networks to be in a state of \textbf{fragile consensus}, meaning that if small number of edges are removed between $S$ and $T$ --- for example by a network administrator --- then polarization will rapidly increase. This is the case even when edges between $S$ and $T$ are eliminated \emph{randomly}, as eliminating edges randomly produces a new $G'$ also drawn from an SBM, but with parameter $q' < q$. Referring to  Figure \ref{fig:fragile_concensus} and Theorem \ref{thm:main_sbm}, $G'$ can have significantly higher polarization than $G$, even when $q'$ is close to $q$.
	

	\subsection{Expectation Analysis}
	\label{sec:exp_anal}
	To prove Theorem \ref{thm:main_sbm}, we apply McSherry's ``perturbation'' approach for analyzing the stochastic block model \cite{McSherry:2001,Spielman:2015}. 
We first bound the polarization of an SBM graph \emph{in expectation}, and then show that the bound carries over to random SBM graphs.

\begin{lemma}
\label{lem:main_thm_expectation}
Let $\overline{G}$ be a graph with $2n$ vertices and adjacency matrix
\begin{equation*}
	\overline{A} = 
	\tikz [baseline=(m.center)]
	\matrix (m) [
	matrix of math nodes,
	left delimiter={[},
	right delimiter={]},
	text depth=1ex,
	text height=3ex,
	text width=.5em,
	anchor=center,
	nodes={minimum width=1.8em, minimum height=1.8em, inner sep=0pt},
	r/.style={fill=red!20},
	b/.style={fill=blue!20}
	] {
		0       & |[r]| p       & |[r]| \ldots& |[r]| p &  |[b]| q & |[b]| q & |[b]| \ldots& |[b]| q  \\
		|[r]| p       & 0       & |[r]| \hspace{-.2em}\ddots & |[r]| \vdots &  |[b]| q & |[b]| q & |[b]| \hspace{-.2em}\ddots & |[b]| \vdots \\
		|[r]| \vdots  & |[r]| \hspace{-.2em}\ddots    & \hspace{-.2em}\ddots  & |[r]| p &  |[b]| \vdots & |[b]| \hspace{-.2em}\ddots  & |[b]| \hspace{-.2em}\ddots & |[b]| q \\
		|[r]| p  & |[r]| \ldots   & |[r]| p &  0 &  |[b]| q & |[b]| \ldots & |[b]| q& |[b]| q \\
		|[b]| q       & |[b]| q       & |[b]| \ldots& |[b]| q &   0 & |[r]| p & |[r]| \hspace{-.2em}\ldots& |[r]| p  \\
		|[b]| q       & |[b]| q       & |[b]| \hspace{-.2em}\ddots  & |[b]| \vdots &  |[r]| p &  0 & |[r]| \ddots& |[r]| \vdots \\
		|[b]| \vdots  & |[b]|  \hspace{-.2em}\ddots    & |[b]| \hspace{-.2em}\ddots  & |[b]| q &  |[r]| \vdots & |[r]| \hspace{-.2em}\ddots  &  \hspace{-.2em}\ddots & |[r]| p \\
		|[b]| q  & |[b]| \ldots   & |[b]| q & |[b]| q &  |[r]| p & |[r]| \ldots & |[r]| p& 0 \\
	}; 
	\end{equation*}
Let $\bv{s}$, as defined in Equation \eqref{eq:our_s}, be the innate opinion vector for the network, and let $\bv{w}^*$ be the resulting equilibrium opinion vector according to the FJ dynamics. Then,
\begin{equation}
	\mathcal{P}_{\bv{w}^*} = \frac{2n}{(2nq + 1)^2}.
\end{equation}
\end{lemma}

\begin{proof}
Let $\overline{D}$ and $\overline{L}$ be the diagonal degree matrix and Laplacian of $\overline{G}$, respectively. Since $\bv{s}$ is mean centered, we have that $\mathcal{P}_{\bv{w}^*} = \bv{s}^T(\overline{L}+I)^{-2}\bv{s}$. To analyze $\mathcal{P}_{\bv{w}^*}$, we need to obtain an explicit representation for the eigendecomposition of $(\overline{L}+I)^{-2}$. 

Let $U = [\bv{u}^{(1)}, \bv{u}^{(2)}]$ where $\bv{u}^{(1)} = \frac{1}{\sqrt{2n}}\bv{1}_{2n}$ and $\bv{u}^{(2)} =  \frac{1}{\sqrt{2n}}\bv{s}$.
We can check that $\overline{A} + p I = U\Lambda U^T$ where $\Lambda = \diag(n(p+q) ,  n(p-q))$.
Now, let ${\overline U} = [\bv{u}^{(1)}, \bv{u}^{(2)}, Z] \in \mathbb{R}^{2n\times 2n}$ where $Z\in \R^{2n\times (2n - 2)}$ is a matrix with orthonormal columns satisfying $Z^T\bv{u}^{(1)} = \bv{0}$ and $Z^T\bv{u}^{(2)} = \bv{0}$.
Such a $Z$ can be obtained by extending $\bv{u}^{(1)}, \bv{u}^{(2)}$ to an orthonormal basis. Note that  $\overline{U}$ is orthogonal, i.e. $\overline{U}\overline{U}^T = \overline{U}^T\overline{U} = I$. 
Since $\overline{L}+ I = (1 + n(p+q))I - (\overline{A} + p I)$, we see that 
	\begin{align}
	\label{eq:decomp_of_L}
	&\overline{L}+I = \overline{U} S \overline U^T.
	\end{align}
	where $S = \diag([1, 2nq + 1, n(p+q) + 1, \ldots, n(p+q) + 1])$.
	
	Since $\overline{U}$ is orthonormal, it follows that $(L+I)^{-2}$ has eigendecomposition $(L+I)^{-2} = \tilde{U} S^{-2} \tilde U^T$ is . Moreover, since $\bv{s}= \sqrt{2n} \cdot\bv{u}^{(2)}$, we have that $\bv{s}$ is orthogonal to $\bv{u}^{(1)}$ and the columns of $Z$. Thus,
	\begin{equation*}
	\mathcal{P}_{\bv{w}^*} = \bv{s}^T(\overline{L}+I)^{-2} \bv{s} = \frac{2n}{(2nq + 1)^2}. \qedhere
	\end{equation*}
\end{proof}

	

	\subsection{Perturbation Analysis}
	With the proof of Lemma \ref{lem:main_thm_expectation} in place, we prove Theorem \ref{thm:main_sbm} by appealing to the following standard result on matrix concentration.
	\begin{lemma}[Corollary of Theorem 1.4 in \cite{Vu:2007}]
		\label{lem:sbm_cor}
		Let $A$ be the adjacency matrix of a graph drawn from the SBM, and let $\overline{A} = \E[A]$ as in Lemma  \ref{lem:main_thm_expectation}. There exists a universal constant $c$ such that if $p \geq c\log^4 n/n$, then with probability $99/100$,
		\begin{align*}
		\|A - \overline{A}\|_2 \leq 3\sqrt{pn}.
		\end{align*}
	\end{lemma}
	We also require a standard Bernstein inequality (see e.g. \cite{Wainwright:2019}):
	\begin{lemma}[Bernstein Inequality]
		\label{lem:bernstein}
		Let $X_i, \ldots, X_m$ be independent random variables with variances $\sigma_1^2, \ldots, \sigma_m^2$ and $|X_i| \leq 1$ almost surely for all $i$. Let ${X} = \sum_{i=1}^m X_i$, $\mu = \E[X]$, and $\sigma^2 = \sum_{i=1}^m \sigma_i^2$. Then, we have the following inequality:
		\begin{align*}
			\Pr[|X - \mu| > \epsilon] \leq e^{\frac{\epsilon^2}{2\sigma^2 + 2\epsilon/3}}.
		\end{align*}
	\end{lemma}
	Using these two bounds, we first prove the following lemma.
	\begin{lemma}
		\label{lem:bound_on_L}
		Let $L$ be the Laplacian of a graph $G$ drawn from the SBM and let $\overline{L} = \E[L]$. For fixed constant $c_0$, with probability 98/100, 
		\begin{align*}
			\|L - \overline{L}\|_2 \leq c_0\sqrt{pn\log n}.
		\end{align*}	
		Note that when $p \geq c \log^4 n$, $c_0\sqrt{pn\log n} \leq \frac{c_0}{\sqrt{c}\log^{1.5} n} \cdot pn$, so for sufficiently large $n$, this lemma implies that $\|L - \overline{L}\|_2 \leq \frac{1}{2} pn$.
	\end{lemma}
	\begin{proof}
		Let $D$ be the degree matrix of $G$ and recall that $\E[D]=\overline{D}$.	
		By  triangle inequality, $\|L - \overline{L}\|_2 \leq \|D - \overline{D}\|_2 +  \|A - \overline{A}\|_2$.
		 By Lemma \ref{lem:sbm_cor}, $\|A-\overline{A}\|_2 < 3\sqrt{pn}$. Additionally, $\|D - \overline{D}\|_2$ is bounded by $\max_{i} |D_{ii} - \overline{D}_{ii}|$. $D_{ii}$ is a sum of Bernoulli random variables with total variance $\sigma^2$ upper bounded by $2np$. It follows from Lemma \ref{lem:bernstein} and our assumption that $p = \Omega(1/n)$ that for any $i$, $|D_{ii} - \overline{D}_{ii}| \leq 
		 c_1\sqrt{pn\log n}$ with probability $1 - \frac{1}{200n}$  for a fixed universal constant $c_1$. By a union bound, we have that $\max_{i} |D_{ii} - \overline{D}_{ii}| \leq \sqrt{pn\log n}$ with probability $99/100$ for all $i$. 
		 A second union bound with the event that $\|A-\overline{A}\|_2 < 3\sqrt{pn}$ gives the lemma with $c_0 = 3 + c_1$.
	\end{proof}

With Lemma \ref{lem:bound_on_L} in place, we are ready to prove Theorem \ref{thm:main_sbm}.
	
	\begin{proof}[Proof of Theorem \ref{thm:main_sbm}]
		We separately consider two cases.
		
		\medskip\noindent \textbf{Case 1, $\mathbf{q \geq p/2}$.}
		In this setting, all eigenvalues of $\overline{L} + I$ lie between $pn + 1$ and $1.5pn + 1$, except for the smallest eigenvalue of $1$, which has corresponding eigenvector $\bv{u}^{(1)} = \bv{1}/\sqrt{2n}$. Since $L\bv{u}^{(1)} = \bv{0}$, $\bv{u}^{(1)}$  is also an eigenvector of $L+I$.
		Let $P = \bv{u}^{(1)}{\bv{u}^{(1)T}}$ be a projection onto this eigenvector. Using that $\bv{u}^{(1)}$ is an eigenvalue of both $L$ and $\overline{L}$ and applying Lemma \ref{lem:bound_on_L}, we have:
		\begin{align*}
		\left(0.5pn + 1\right) (I - P) \preceq (I - P)(L+I)(I - P) \preceq (2pn + 1) (I - P).
		\end{align*}
		Since $(I - P)(L+I)(I - P)$ and $(I - P)$ commute, it follows that $(0.5 pn + 1)^2 (I - P) \preceq (I-P)(L + I)^2(I-P) \preceq (2pn + 1)^2(I - P)$. Finally, noting that $(I - P)\bv{s} = \bv{s}$, $\bv{s}^T\bv{s} = 2n$, and $M \preceq N \Rightarrow N^{-1} \preceq M^{-1}$ gives the Theorem for $q \geq p/2$.

	\medskip\noindent \textbf{Case 2, $\mathbf{q < p/2}$.} The small $q$ case is more challenging, requiring a strengthening of Lemma \ref{lem:bound_on_L}. Lemma \ref{lem:bound_on_L} asserts that every eigenvalue of $L$ is within additive error $c_0\sqrt{pn\log n}$ from the corresponding eigenvalue in $\overline{L}$. While strong for $\overline{L}$'s largest eigenvalues of $(p+q)n$, the statement can be weak for  $L$'s smallest non-zero eigenvalue of $2nq$. We require a tighter \emph{relative error} bound, which we show in the lemma below.
	
		\begin{lemma}\label{lem:small_eig_bound}
		Assume $1/n \leq q < p/2$. Let $\lambda_{2}(L) $ be $L$'s smallest non-zero eigenvalue. With probability $99/100$, for sufficiently large $n$,
		\begin{align*}
		\frac{1}{2}nq \leq \lambda_{2}(L)  \leq 4 nq.
		\end{align*}
	\end{lemma}
	
	Before proving Lemma \ref{lem:small_eig_bound}, we use it to complete the proof of Theorem \ref{thm:main_sbm}. To do so, we require the following well-known bound.
	\begin{lemma}[Davis-Kahan Theorem \cite{DavisKahan:1970}]
		\label{lem:dk}
		Let $M$ and $H$ be $m\times m$ symmetric matrices with eigenvectors $\bv{v}_1, \ldots, \bv{v}_m$ and $\bv{\tilde v}_1, \ldots, \bv{\tilde v}_m$, respectively, and eigenvalues $\lambda_1, \ldots, \lambda_m$ and $\tilde{\lambda}_1, \ldots, \tilde{\lambda}_m$. 
		If $\|M - H\|_2 \leq \epsilon$, then for all $i$,
		\begin{align*}
		(\bv{v}_i^T\bv{\tilde v}_i)^2 \geq 1 - \frac{\epsilon^2}{\min_{j \neq i} |\lambda_i - \lambda_j|^2}.
		\end{align*}
	\end{lemma}
		Continuing the proof of Theorem \ref{thm:main_sbm}, let $P = \bv{u}^{(1)}{\bv{u}^{(1)T}}$. Let $\tilde{U} \in \R^{2n \times (2n-1)}$ be an orthonormal basis for the span of $I-P$. We will apply Lemma  \ref{lem:dk} to the matrices $\tilde{U}^T \overline{L} \tilde{U}$ and $\tilde{U}^T L \tilde{U}$. Since $\bv{u}^{(1)}$ is an eigenvector of both $\overline{L}$ and $L$, the eigenvectors of $\tilde{U}^T \overline{L} \tilde{U}$ and $\tilde{U}^T L \tilde{U}$ are equal to the remaining $2n-1$ eigenvectors of $\overline{L}$ and $L$ left multiplied by $\tilde{U}^T$. 
		The eigenvalues of $\tilde{U}^T \overline{L} \tilde{U}$ and $\tilde{U}^T L \tilde{U}$ are simply the non-zero eigenvalues of $\overline{L}$ and $L$. 

		Let  $\bv{y}$ be the eigenvector of $L$ associated with $\lambda_{2}(L)$. 
		Theorem \ref{lem:bound_on_L} implies that $\|\tilde{U}^T \overline{L} \tilde{U} - \tilde{U}^T L \tilde{U}\| \leq c_0\sqrt{pn\log n}$ and so by Lemma \ref{lem:dk}, we have:
		\begin{align*}
		(\bv{u}^{(2)T}\tilde{U}^T\tilde{U}\bv{y})^2 \geq 1 - \frac{c_0^2pn\log n}{((p-q)n)^2}.
		\end{align*}
		Since $p-q \geq p/2$, our assumption that $p = \Omega(\log^4n/n)$ implies that $(\bv{u}^{(2)T}\tilde{U}^T\tilde{U}\bv{y})^2 \geq 1 - O(1/\log^3 n))$, which is $\geq 1/2$ for large enough $n$.
		Since $\bv{y}$ and $\bv{u}^{(2)}$ are eigenvalues of $L$ and $\overline{L}$ respectively, both are orthogonal to $\bv{u}^{(1)}$. So $(\bv{u}^{(2)T}\tilde{U}^T\tilde{U}\bv{y})^2 = (\bv{u}^{(2)T}\bv{y})^2$. We conclude:
		\begin{align}\label{eq:inner_product_bound}
		1/2 \leq (\bv{y}^T\bv{u}^{(2)})^2 \leq 1.
		\end{align}
		In other words, $L$'s second eigenvector $\bv{y}$ has a large inner product with $\overline{L}$'s second eigenvector $\bv{u}^{(2)}$.
	
Since $L$ and $(L+I)^{-2}$ have the same eigenvectors, we can bound $\pP_{\bv{z}^*} = \bv{s}^T(L+I)^{-2}\bv{s} = 2n \cdot \bv{u}^{(2)T}(L+I)^{-2}\bv{u}^{(2)}$ as follows:
\begin{align*}
&(\bv{y}^T\bv{u}^{(2)})^2 (\lambda_{2}(L) +1)^{-2} \leq \frac{1}{2n}\pP_{\bv{z}^*} \\ &\leq (\bv{y}^T\bv{u}^{(2)})^2 (\lambda_{2}(L) +1)^{-2} +  (1-(\bv{y}^T\bv{u}^{(2)})^2)\|(L+I)^{-2}R\|_2
\end{align*}
where $R = I - \bv{y}\bv{y}^T - \bv{u}^{(1)} \bv{u}^{(1)T}$ is a projection matrix onto $(L+I)$'s largest $n-2$ eigenvectors. From the same argument used for Case 1, all of these eigenvectors have corresponding eigenvalues $\geq \frac{1}{2}pn + 1$, and thus $\|(L+I)^{-2}R\|_2 \leq \frac{1}{(\frac{1}{2}pn + 1)^2} \leq \frac{1}{(qn + 1)^2}$.
Applying \eqref{eq:inner_product_bound} and Lemma \ref{lem:small_eig_bound}, we have:
\begin{align*}
\frac{1}{2} \frac{2n}{(4nq +1)^2} \leq \pP_{\bv{z}^*} \leq \frac{3n}{(\frac{1}{2}nq +1)^2},
\end{align*} 
which establishes the theorem.
\end{proof}		

All we have left is to prove Lemma \ref{lem:small_eig_bound}.
	\begin{proof}[Proof of Lemma \ref{lem:small_eig_bound}]
		We will first prove that 
		\begin{align}
		\label{eq:upper_bound}
		\lambda_{2}(L) \leq 4nq.
		\end{align}
		To do so, we apply the Courant-Fischer minmax theorem, from which is suffices to exhibit two orthogonal unit vectors with Rayleigh quotient $\leq 4nq$. The first will be $\bv{u}^{(1)} = \bv{1}/\sqrt{2n}$, which has Rayleigh quotient $\bv{u}^{(1)T}L\bv{u}^{(1)} = 0$. 
		The second will be $\bv{u}^{(2)} = \bv{s}/\sqrt{2n}$, which is orthogonal to $\bv{u}^{(1)}$.
		
		Let $\gamma = \bv{u}^{(2)T} L\bv{u}^{(2)}$.
		We can check that $2n\gamma$ is exactly equal to the number of entries in the off diagonal $n \times n$ blocks of $A$ -- i.e. it is twice the number of edges in $G$ \emph{between different communities}. Appealing to Lemma \ref{lem:bernstein}, we have that $|2n^2q - 2n\gamma| \leq 11n\sqrt{q}$ with probability $99/100$ as long as $q > 1/n$. So for large enough $n$,
		\begin{align}\label{eq:123}
			\frac{1}{2} nq \leq \gamma \leq 4nq,
		\end{align} which proves \eqref{eq:upper_bound}.
		Next we establish that 
		\begin{align}
		\label{eq:lower_bound}
		\lambda_{2}(L) \geq \frac{1}{2}nq.
		\end{align}
		Again we apply the minmax principle, this time exhibiting $2n - 1$ vectors with Rayleigh quotient $\geq \frac{1}{2}nq$. The first vector will be $\bv{u}^{(2)}$, which we know has sufficiently large Rayleigh quotient by \eqref{eq:123}. For the remaining vectors, let $\bv{z}_1, \ldots, \bv{z}_{2n - 2} \in \R^{2n}$ be a set of $2n - 2$ orthonormal vectors which are orthogonal to both $\bv{u}^{(1)}$ and $\bv{u}^{(2)}$. 
		
		From \eqref{eq:inner_product_bound}, we can derive that for any $i$, $(\bv{z}_i^T \bv{y})^2 \leq 1/2$, where $\bv{y}$ is the eigenvector corresponding to $L$'s second smallest eigenvalue. Since all other eigenvalues of $L$ are $\geq \frac{1}{2}pn$, we thus have, for all $i$, 
		\begin{align*}
		\bv{z}_i^TL \bv{z}_i \geq \frac{1}{2}\cdot \frac{1}{2}pn \geq \frac{1}{2}qn.
		\end{align*}
		This proves \eqref{eq:lower_bound}, which along with \eqref{eq:upper_bound} gives the Lemma.
	\end{proof}

	%
	
	
	\begin{figure*}
		\begin{subfigure}{0.5\textwidth}
			\centering
			\includegraphics[height=4.5cm]{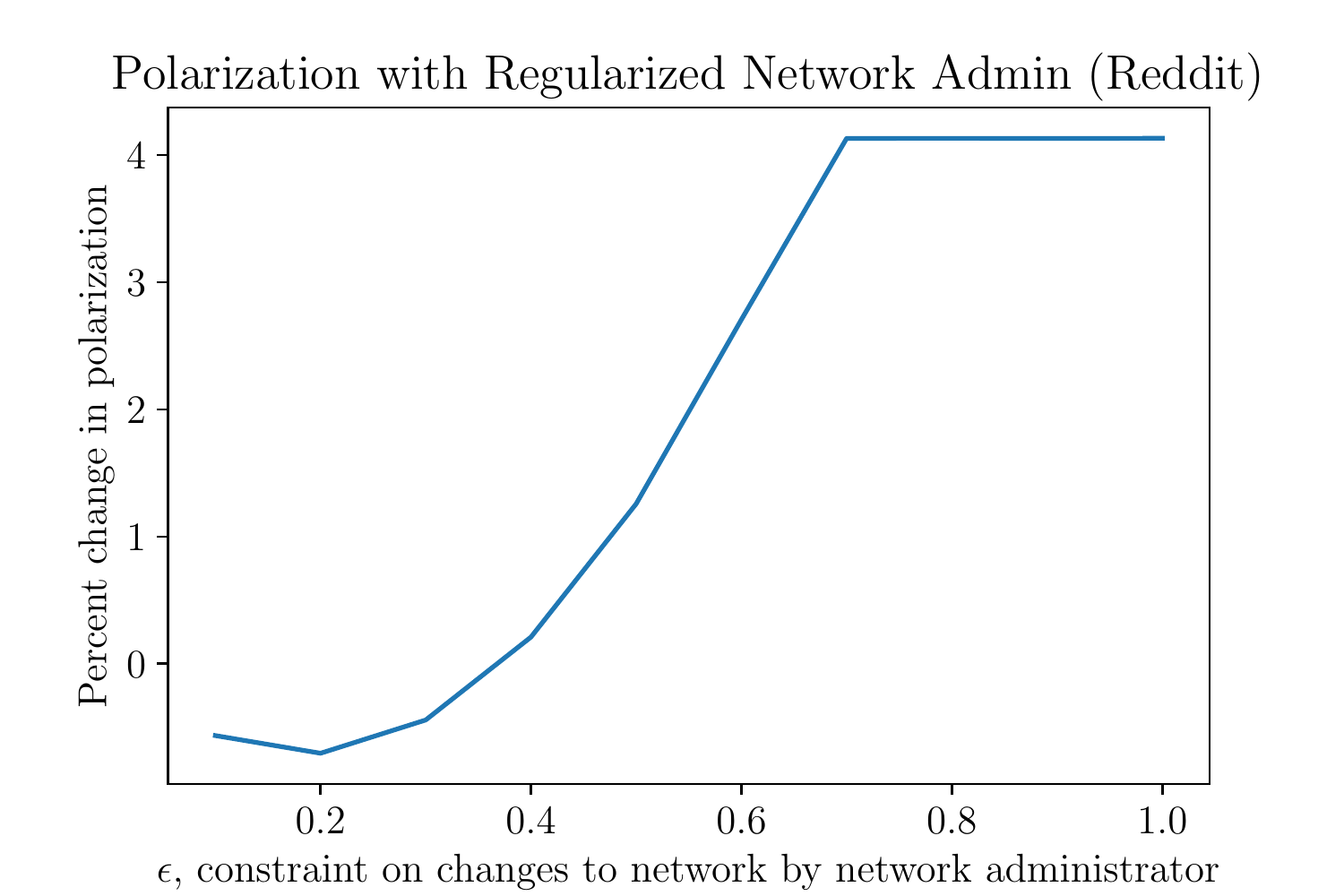}
			\caption{Change in polarization vs $\epsilon$ for Reddit network.}  \label{subfig:reddit_pol_fix}
		\end{subfigure}
		\hfill
		\begin{subfigure}{0.5\textwidth}
			\centering
			\includegraphics[height=4.5cm]{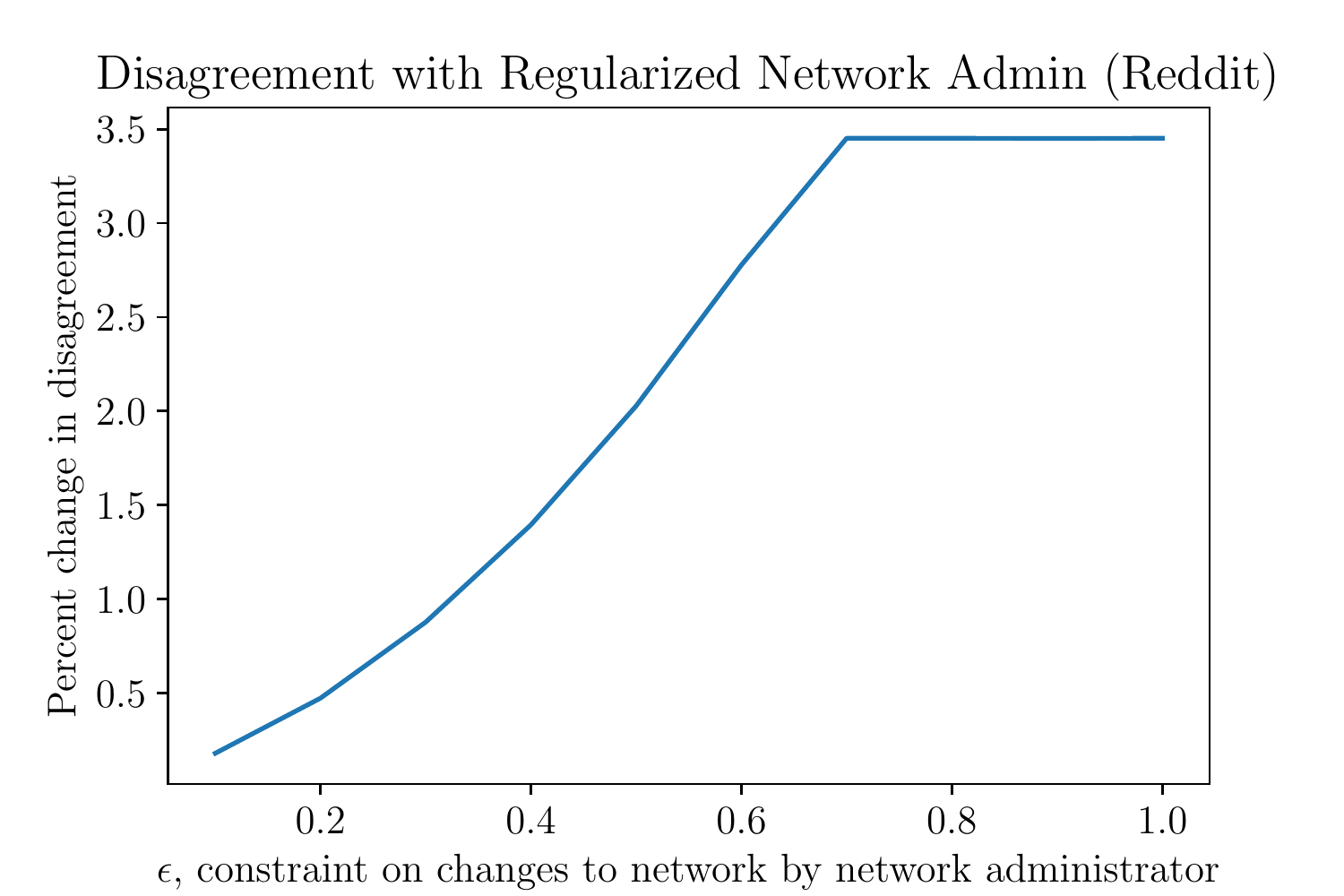}
			\caption{Change in disagreement vs $\epsilon$ for Reddit network.}  \label{subfig:reddit_disagg_fix}
		\end{subfigure}
		\hfill
		\begin{subfigure}{0.5\textwidth}
			\centering
			\includegraphics[height=4.5cm]{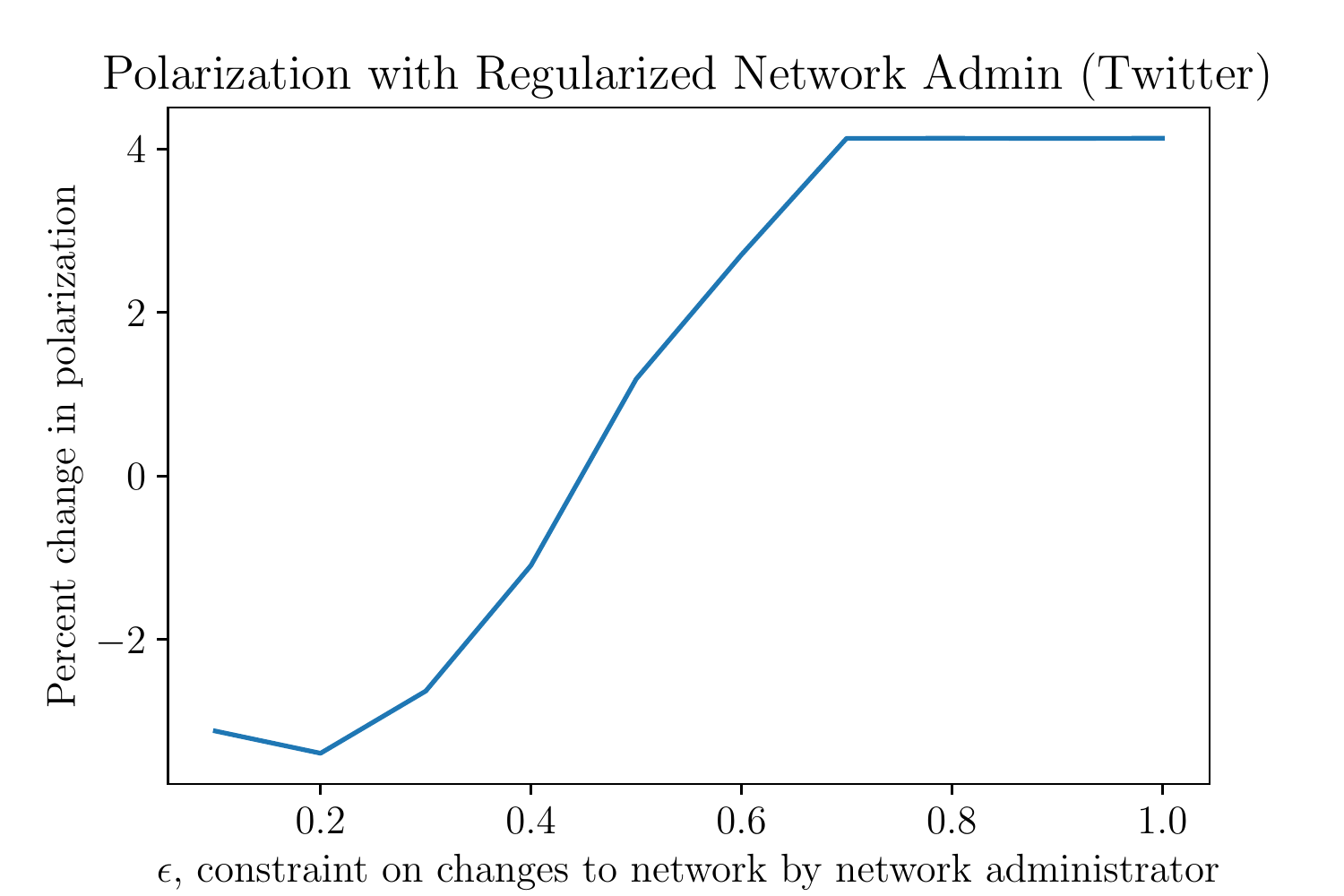}
			\caption{Change in polarization vs $\epsilon$ for Twitter network.}  \label{subfig:twitter_pol_fix}
		\end{subfigure}
		\hfill
		\begin{subfigure}{0.5\textwidth}
			\centering
			\includegraphics[height=4.5cm]{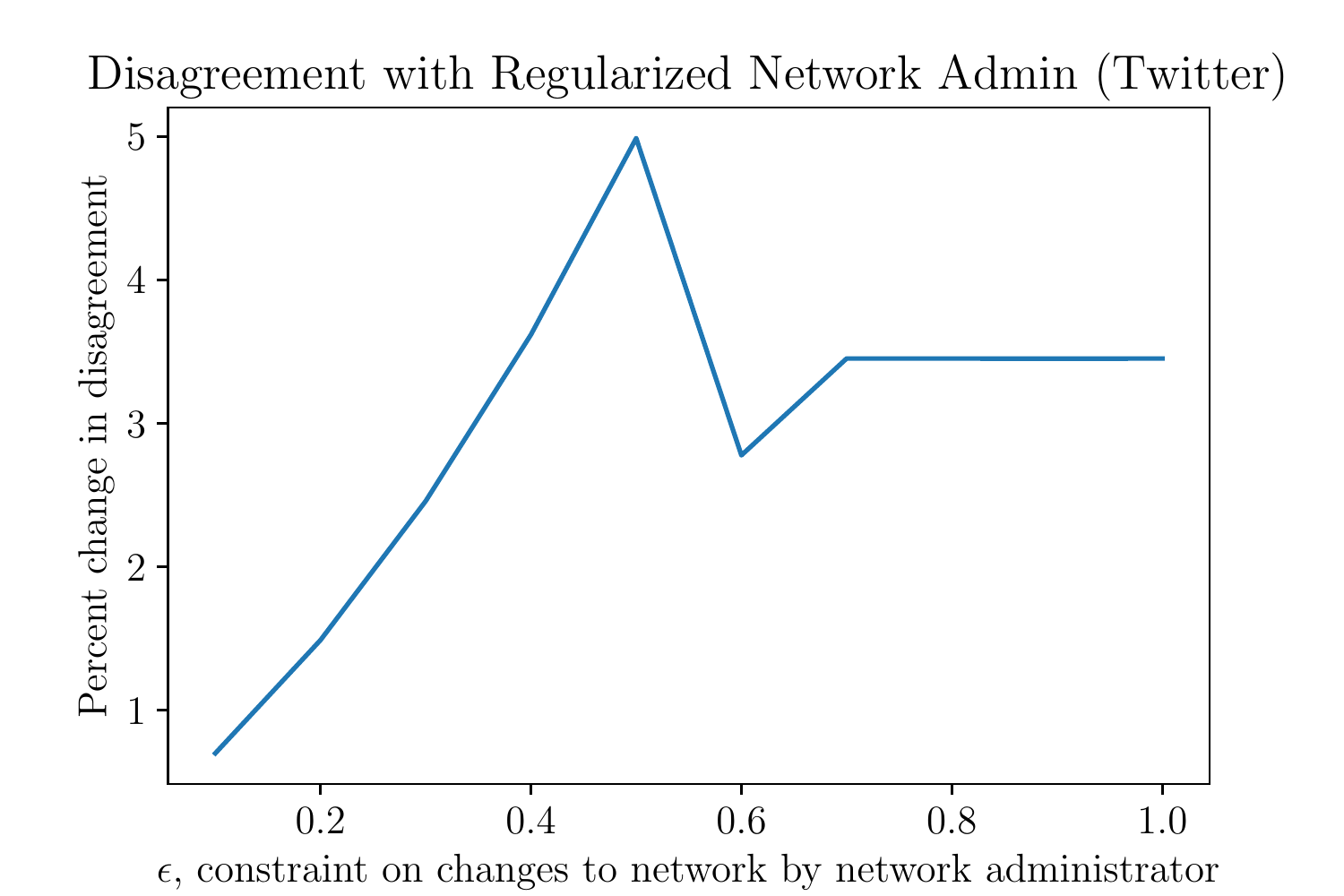}
			\caption{Change in disagreement vs $\epsilon$ for Twitter network.}  \label{subfig:twitter_disagg}
		\end{subfigure}
		\caption{Applying \emph{regularized} network administrator dynamics to real-world social networks, $\gamma=0.2$. Details in Section \ref{sec:remedy}.}   \label{fig:reddit_twitter_fix}
	\end{figure*}
	
	
	\section{A Simple Remedy}
	\label{sec:remedy}
	
	Throughout this paper, our results have largely been pessimistic. When we introduce the network administrator, an external actor who filters content for users, we see that polarization rises and echo chambers form, along the lines of Pariser's filter bubble theory \cite{Pariser:2011}. Our analysis in the SBM further evidences that social networks can easily be in a state of ``fragile consensus", which leaves them vulnerable to become extremely polarized even when only a small number of edges are modified.
	
	In this section, however, we conclude with a positive result. We find that, with a slight modification to the network administrator dynamics, the filter bubble effect is vastly mitigated. Even more surprisingly, disagreement also barely increases, showing that it is possible for the network administrator to reduce polarization in the network while not hurting its own objective.
	
	\subsection{Regularized Dynamics}
	
	We modify the role of the network administrator by adding an $L^2$ regularization term to its objective function.
	
\begin{tcolorbox}[pad at break=1mm]
		\textbf{Regularized Network Administrator Dynamics.} \\
		Given initial graph $G^{(0)} = G$ and initial opinions $z^{(0)} = s$, in each round $r=1, 2, 3, \dots$
		\begin{itemize}
			\item 
			First, the users adopt new expressed opinions $z^{(r)}$. These opinions are the equilibrium opinions (Equation \ref{eq:fj_solution}) of the FJ dynamics model applied to $G^{(r-1)}$:
			\begin{equation}
			\label{eq:na_dyn_z}
			z^{(r)} = (L^{(r-1)}+I)^{-1}s.
			\end{equation}
			Here $L^{(r-1)}$ is the Laplacian of $G^{(r-1)}$.
			\item 
			Then, given user opinions $z^{(r)}$, the network administrator minimizes disagreement by modifying the graph, subject to certain restrictions:
			\begin{equation} \label{eq:na_game_update_w}
			G^{(r)} = \argmin_{G \in S}\; \dD_{G, z^{(r)}} + + \gamma \|W\|_F^2
			\end{equation}
			$S$ is the constrained set of graphs the network administrator is allowed to change to, $W$ is the adjacency matrix of $G$, and $\gamma > 0$ is a fixed constant.
		\end{itemize}
	\end{tcolorbox}
	
	$\gamma > 0$ is a fixed constant that controls the strength of regularization. We use $L^2$ regularization because $\displaystyle\left|\argmin_{x : \|x\|_1=1} \|x\|_2\right| = \bv{1}_n/n$ for $x \in \mathbb{R}^n$.
	%
	So intuitively, since the network administrator must keep the total edge weight of the graph constant, the addition of the regularization term encourages the network administrator to make modifications to many edges in the graph, instead of making large, concentrated changes to a small number of edges. 
	
	\subsection{Results}
	Figure \ref{fig:reddit_twitter_fix} shows the results of the \emph{regularized} network administrator dynamics on the Reddit and Twitter networks, with $\gamma=0.2$. Polarization increases by at most $4\%$, no matter the value of $\epsilon$. This is a drastic difference from the non-regularized network administrator dynamics, where polarization increased by over $4000\%$. Disagreement, which the network administrator is incentivized to decrease, increases by at most $5\%$.
	
	
	\section{Conclusion and Future Directions}
	\label{sec:open}
	
	Social media has become an integral part of our lives, with a majority of people using at least one social media app daily \cite{SmithAndersen2018}. Despite enabling users access to a diversity of information, social media usage has been linked to increased societal polarization \cite{FlaxmanGoelRao:2016}. One popular theory to explain this phenomenon is the idea of \emph{filter bubbles}: by automatically recommending content that a user is likely to agree with (i.e. content filtering), social network algorithms create polarized ``echo chambers" of users \cite{Pariser:2011}. While intellectually satisfying, evidence for the filter bubble theory is mainly post-hoc, and the theory has been met with skepticism \cite{Boutin:2011,NguyenHuiHarper:2014,Barbera2014,VaccariValerianiBarberaJostNagler2016,Zuiderveen-BorgesiusTrillingMoeller:2016,HosanagarFlederLeeBuja2014}.
	
	In this work, we propose an extension to the Friedkin-Johnsen opinion dynamics model that explicitly models recommendation systems in social networks. Using this model, we experimentally show the emergence of filter bubbles in real-world networks. We also provide theoretical justification for why social networks are so vulnerable to outside actors.
	
	Our work poses many follow-up questions. For example, as discussed earlier, variants of the Bounded Confidence Model (BCM) have previously been used to argue that polarization is caused by ``biased assimilation" of content by users \cite{DandekarGoelLee:2013, Del-VicarioScalaCaldarelli:2017,GeschkeLorenzHoltz:2019}. In this work, we use the Friedkin-Johnsen opinion dynamics model because of its linear algebraic interpretation, which allows us to establish concrete theoretical results. An interesting follow-up would be to incorporate our network administrator dynamics into the more complex BCM variants used by \cite{GeschkeLorenzHoltz:2019} and \cite{Del-VicarioScalaCaldarelli:2017}.
	
	Another interesting direction is modeling the interference of other outside actors, as our theoretical analysis is not limited to recommendation systems. Can we develop a similar framework for modeling the effects of cyber warfare (see e.g. \cite{RileyRobertson2017}) on societal polarization? And perhaps more importantly, can we also develop methods to mitigate the effects of cyber warfare on polarization?
	
	\bibliography{kdd_workshop}

\begin{thebibliography}{64}
\providecommand{\natexlab}[1]{#1}
\providecommand{\url}[1]{\texttt{#1}}
\expandafter\ifx\csname urlstyle\endcsname\relax
  \providecommand{\doi}[1]{doi: #1}\else
  \providecommand{\doi}{doi: \begingroup \urlstyle{rm}\Url}\fi

\bibitem[Abbe(2018)]{Abbe:2018}
Emmanuel Abbe.
\newblock Community detection and stochastic block models.
\newblock \emph{Foundations and Trends in Communications and Information
  Theory}, 14\penalty0 (1-2):\penalty0 1--162, 2018.

\bibitem[Abebe et~al.(2018)Abebe, Kleinberg, Parkes, and
  Tsourakakis]{AbebeKleinbergParkes:2018}
Rediet Abebe, Jon Kleinberg, David Parkes, and Charalampos~E Tsourakakis.
\newblock Opinion dynamics with varying susceptibility to persuasion.
\newblock In \emph{KDD 2018}, pages 1089--1098. ACM, 2018.

\bibitem[Adamic and Glance(2005)]{AdamicGlance:2005}
Lada~A Adamic and Natalie Glance.
\newblock The political blogosphere and the 2004 {US} election: divided they
  blog.
\newblock In \emph{Proceedings of the 3rd International Workshop on Link
  Discovery (LinkKDD)}, pages 36--43, 2005.

\bibitem[Aggarwal(2016)]{Aggarwal2016}
Charu~C. Aggarwal.
\newblock \emph{Recommender Systems: The Textbook}.
\newblock Springer Publishing Company, Incorporated, 1st edition, 2016.
\newblock ISBN 3319296574, 9783319296579.

\bibitem[Aslay et~al.(2018)Aslay, Matakos, Galbrun, and
  Gionis]{AslayMatakosGalbrun:2018}
Cigdem Aslay, Antonis Matakos, Esther Galbrun, and Aristides Gionis.
\newblock Maximizing the diversity of exposure in a social network.
\newblock In \emph{2018 IEEE International Conference on Data Mining (ICDM)},
  pages 863--868. IEEE, 2018.

\bibitem[Baer(2016)]{Baer:2016}
Drake Baer.
\newblock The {'Filter Bubble'} explains why {Trump} won and you didn't see it
  coming.
\newblock \emph{Science of Us}, 2016.

\bibitem[Barber{\'a}(2014)]{Barbera2014}
Pablo Barber{\'a}.
\newblock How social media reduces mass political polarization. evidence from
  germany, spain, and the us.
\newblock \emph{Job Market Paper, New York University}, 46, 2014.

\bibitem[Becchetti et~al.(2017)Becchetti, Clementi, Natale, Pasquale, and
  Trevisan]{BecchettiClementiNatale:2017}
Luca Becchetti, Andrea Clementi, Emanuele Natale, Francesco Pasquale, and Luca
  Trevisan.
\newblock Find your place: Simple distributed algorithms for community
  detection.
\newblock In \emph{SODA 2017}, pages 940--959, 2017.

\bibitem[Beck(2015)]{Beck2015}
A.~Beck.
\newblock On the convergence of alternating minimization for convex programming
  with applications to iteratively reweighted least squares and decomposition
  schemes.
\newblock \emph{SIAM Journal on Optimization}, 25\penalty0 (1):\penalty0
  185--209, 2015.
\newblock \doi{10.1137/13094829X}.
\newblock URL \url{https://doi.org/10.1137/13094829X}.

\bibitem[Bertsekas(1999)]{Bertsakas1999}
D.P. Bertsekas.
\newblock \emph{Nonlinear Programming}.
\newblock Athena Scientific, 1999.

\bibitem[Bindel et~al.(2015)Bindel, Kleinberg, and
  Oren]{BindelKleinbergOren:2015}
David Bindel, Jon Kleinberg, and Sigal Oren.
\newblock How bad is forming your own opinion?
\newblock \emph{Games and Economic Behavior}, 92:\penalty0 248 -- 265, 2015.
\newblock ISSN 0899-8256.
\newblock \doi{https://doi.org/10.1016/j.geb.2014.06.004}.
\newblock URL
  \url{http://www.sciencedirect.com/science/article/pii/S0899825614001122}.

\bibitem[Binder(2014)]{Binder:2014}
Sarah Binder.
\newblock Polarized we govern?, 2014.

\bibitem[Boggs and Tolle(1995)]{BoggsTolle1995}
Paul~T. Boggs and Jon~W. Tolle.
\newblock Sequential quadratic programming.
\newblock \emph{Acta Numerica}, 4:\penalty0 1--51, 1995.
\newblock \doi{10.1017/S0962492900002518}.

\bibitem[Boutin(2011)]{Boutin:2011}
Paul Boutin.
\newblock Your results may vary: Will the information superhighway turn into a
  cul-de-sac because of automated filters?
\newblock \emph{The Wall Street Journal}, 2011.

\bibitem[Brundidge(2010)]{Brundidge:2010}
Jennifer Brundidge.
\newblock Encountering ``difference'' in the contemporary public sphere: The
  contribution of the internet to the heterogeneity of political discussion
  networks.
\newblock \emph{Journal of Communication}, 60\penalty0 (4):\penalty0 680--700,
  11 2010.

\bibitem[Carletti et~al.(2006)Carletti, Fanelli, Grolli, and
  Guarino]{CarlettiFanelliGrolli:2006}
T.~Carletti, D.~Fanelli, S.~Grolli, and A.~Guarino.
\newblock How to make an efficient propaganda.
\newblock \emph{Europhysics Letters}, 74\penalty0 (2):\penalty0 222, 2006.

\bibitem[Chaney et~al.(2018)Chaney, Stewart, and
  Engelhardt]{ChaneyStewartEngelhardt2018}
Allison J.~B. Chaney, Brandon~M. Stewart, and Barbara~E. Engelhardt.
\newblock How algorithmic confounding in recommendation systems increases
  homogeneity and decreases utility.
\newblock In \emph{Proceedings of the 12th ACM Conference on Recommender
  Systems}, RecSys '18, pages 224--232, New York, NY, USA, 2018. ACM.
\newblock ISBN 978-1-4503-5901-6.
\newblock \doi{10.1145/3240323.3240370}.
\newblock URL \url{http://doi.acm.org/10.1145/3240323.3240370}.

\bibitem[Chen et~al.(2018{\natexlab{a}})Chen, Lijffijt, and
  De~Bie]{ChenLijffijtDe-Bie:2018}
Xi~Chen, Jefrey Lijffijt, and Tijl De~Bie.
\newblock Quantifying and minimizing risk of conflict in social networks.
\newblock In \emph{KDD 2018}, pages 1197--1205, 2018{\natexlab{a}}.

\bibitem[Chen et~al.(2018{\natexlab{b}})Chen, Lijffijt, and
  De~Bie]{ChenLijffijtDe-Bie:2018a}
Xi~Chen, Jefrey Lijffijt, and Tijl De~Bie.
\newblock The normalized friedkin-johnsen model (a work-in-progress report).
\newblock In \emph{ECML PKDD 2018-PhD Forum}, 2018{\natexlab{b}}.

\bibitem[Conover et~al.(2011)Conover, Ratkiewicz, Francisco, Gon{\c{c}}alves,
  Menczer, and Flammini]{ConoverRatkiewiczFrancisco:2011}
Michael~D Conover, Jacob Ratkiewicz, Matthew Francisco, Bruno Gon{\c{c}}alves,
  Filippo Menczer, and Alessandro Flammini.
\newblock Political polarization on twitter.
\newblock In \emph{Fifth international AAAI conference on weblogs and social
  media}, 2011.

\bibitem[Dandekar et~al.(2013)Dandekar, Goel, and Lee]{DandekarGoelLee:2013}
Pranav Dandekar, Ashish Goel, and David~T Lee.
\newblock Biased assimilation, homophily, and the dynamics of polarization.
\newblock \emph{Proceedings of the National Academy of Sciences}, 110\penalty0
  (15):\penalty0 5791--5796, 2013.

\bibitem[Das et~al.(2014)Das, Gollapudi, and
  Munagala]{DasGollapudiMunagala:2014}
Abhimanyu Das, Sreenivas Gollapudi, and Kamesh Munagala.
\newblock Modeling opinion dynamics in social networks.
\newblock In \emph{Proceedings of the 7th ACM International Conference on Web
  Search and Data Mining}, pages 403--412. ACM, 2014.

\bibitem[Davis and Kahan(1970)]{DavisKahan:1970}
C.~Davis and W.~Kahan.
\newblock The rotation of eigenvectors by a perturbation. iii.
\newblock \emph{SIAM Journal on Numerical Analysis}, 7\penalty0 (1):\penalty0
  1--46, 1970.
\newblock \doi{10.1137/0707001}.
\newblock URL \url{https://doi.org/10.1137/0707001}.

\bibitem[De et~al.(2014)De, Bhattacharya, Bhattacharya, Ganguly, and
  Chakrabarti]{De2014}
Abir De, Sourangshu Bhattacharya, Parantapa Bhattacharya, Niloy Ganguly, and
  Soumen Chakrabarti.
\newblock Learning a linear influence model from transient opinion dynamics.
\newblock In \emph{Proceedings of the 23rd ACM International Conference on
  Conference on Information and Knowledge Management}, CIKM '14, pages
  401--410, New York, NY, USA, 2014. ACM.
\newblock ISBN 978-1-4503-2598-1.
\newblock \doi{10.1145/2661829.2662064}.
\newblock URL \url{http://doi.acm.org/10.1145/2661829.2662064}.

\bibitem[DeGroot(1974)]{DeGroot:1974}
Morris~H DeGroot.
\newblock Reaching a consensus.
\newblock \emph{Journal of the American Statistical Association}, 69\penalty0
  (345):\penalty0 118--121, 1974.

\bibitem[Del~Vicario et~al.(2017)Del~Vicario, Scala, Caldarelli, Stanley, and
  Quattrociocchi]{Del-VicarioScalaCaldarelli:2017}
Michela Del~Vicario, Antonio Scala, Guido Caldarelli, H.~Eugene Stanley, and
  Walter Quattrociocchi.
\newblock Modeling confirmation bias and polarization.
\newblock \emph{Scientific reports}, 7, 2017.

\bibitem[Eadicicco(2019)]{Eadicicco:2019}
Lisa Eadicicco.
\newblock Apple {CEO Tim Cook} urges college grads to `push back' against
  algorithms that promote the `things you already know, believe, or like'.
\newblock \emph{Business Insider}, 2019.

\bibitem[Epitropou et~al.(2017)Epitropou, Fotakis, Hoefer, and
  Skoulakis]{EpitropouFotakisHoeferSkoulakis:2017}
Markos Epitropou, Dimitris Fotakis, Martin Hoefer, and Stratis Skoulakis.
\newblock Opinion formation games with aggregation and negative influence.
\newblock In \emph{Algorithmic Game Theory - 10th International Symposium,
  {SAGT} 2017, L'Aquila, Italy, September 12-14, 2017, Proceedings}, pages
  173--185, 2017.

\bibitem[Evans(2018)]{Evans:2018}
Benedict Evans.
\newblock The death of the newsfeed.
\newblock
  https://www.ben-evans.com/benedictevans/2018/4/2/the-death-of-the-newsfeed,
  2018.

\bibitem[Flaxman et~al.(2016)Flaxman, Goel, and Rao]{FlaxmanGoelRao:2016}
Seth Flaxman, Sharad Goel, and Justin~M Rao.
\newblock Filter bubbles, echo chambers, and online news consumption.
\newblock \emph{Public Opinion Quarterly}, 80\penalty0 (S1):\penalty0 298--320,
  2016.

\bibitem[Friedkin and Johnsen(1990)]{FriedkinJohnsen:1990}
Noah~E Friedkin and Eugene~C Johnsen.
\newblock Social influence and opinions.
\newblock \emph{Journal of Mathematical Sociology}, 15\penalty0 (3-4):\penalty0
  193--206, 1990.

\bibitem[Garimella et~al.(2017)Garimella, Gionis, Parotsidis, and
  Tatti]{GarimellaGionisParotsidis:2017}
Kiran Garimella, Aristides Gionis, Nikos Parotsidis, and Nikolaj Tatti.
\newblock Balancing information exposure in social networks.
\newblock In \emph{\NIPS{2017}}, pages 4663--4671. 2017.

\bibitem[Garimella et~al.(2018)Garimella, Morales, Gionis, and
  Mathioudakis]{GarimellaMoralesGionisMathioudakis2018}
Kiran Garimella, Gianmarco De~Francisci Morales, Aristides Gionis, and Michael
  Mathioudakis.
\newblock Quantifying controversy on social media.
\newblock \emph{Trans. Soc. Comput.}, 1\penalty0 (1):\penalty0 3:1--3:27,
  January 2018.
\newblock ISSN 2469-7818.
\newblock \doi{10.1145/3140565}.
\newblock URL \url{http://doi.acm.org/10.1145/3140565}.

\bibitem[Geschke et~al.(2019)Geschke, Lorenz, and
  Holtz]{GeschkeLorenzHoltz:2019}
Daniel Geschke, Jan Lorenz, and Peter Holtz.
\newblock The triple-filter bubble: Using agent-based modelling to test a
  meta-theoretical framework for the emergence of filter bubbles and echo
  chambers.
\newblock \emph{British Journal of Social Psychology}, 58\penalty0
  (1):\penalty0 129--149, 2019.

\bibitem[Gionis et~al.(2013)Gionis, Terzi, and
  Tsaparas]{GionisTerziTsaparas:2013}
Aristides Gionis, Evimaria Terzi, and Panayiotis Tsaparas.
\newblock Opinion maximization in social networks.
\newblock In \emph{Proceedings of the 13th {SIAM} International Conference on
  Data Mining, May 2-4, 2013. Austin, Texas, {USA.}}, pages 387--395, 2013.

\bibitem[Hare and Poole(2014)]{HarePoole:2014}
Christopher Hare and Keith~T Poole.
\newblock The polarization of contemporary american politics.
\newblock \emph{Polity}, 46\penalty0 (3):\penalty0 411--429, 2014.

\bibitem[Hegselmann and Krause(2002)]{HegselmannKrause2002}
Rainer Hegselmann and Ulrich Krause.
\newblock Opinion dynamics and bounded confidence: Models, analysis and
  simulation.
\newblock \emph{Journal of Artificial Societies and Social Simulation},
  5:\penalty0 1--24, 2002.

\bibitem[Hegselmann et~al.(2002)Hegselmann, Krause,
  et~al.]{HegselmannKrauseothers:2002}
Rainer Hegselmann, Ulrich Krause, et~al.
\newblock Opinion dynamics and bounded confidence models, analysis, and
  simulation.
\newblock \emph{Journal of Artificial Societies and Social Simulation},
  5\penalty0 (3), 2002.

\bibitem[Holland et~al.(1983)Holland, Laskey, and
  Leinhardt]{HollandLaskeyLeinhardt:1983}
Paul~W Holland, Kathryn~Blackmond Laskey, and Samuel Leinhardt.
\newblock Stochastic blockmodels: First steps.
\newblock \emph{Social networks}, 5\penalty0 (2):\penalty0 109--137, 1983.

\bibitem[Holone(2016)]{Holone:2016}
Harald Holone.
\newblock The filter bubble and its effect on online personal health
  information.
\newblock \emph{Croatian medical journal}, 57\penalty0 (3):\penalty0 298, 2016.

\bibitem[Hosanagar et~al.(2014)Hosanagar, Fleder, Lee, and
  Buja]{HosanagarFlederLeeBuja2014}
Kartik Hosanagar, Daniel Fleder, Dokyun Lee, and Andreas Buja.
\newblock Will the global village fracture into tribes? recommender systems and
  their effects on consumer fragmentation.
\newblock \emph{Management Science}, 60\penalty0 (4):\penalty0 805--823, 2014.
\newblock \doi{10.1287/mnsc.2013.1808}.
\newblock URL \url{https://doi.org/10.1287/mnsc.2013.1808}.

\bibitem[Jackson(2017)]{Jackson:2017}
Jasper Jackson.
\newblock Eli pariser: activist whose filter bubble warnings presaged {Trump}
  and {Brexit}: {Upworthy} chief warned about dangers of the internet's echo
  chambers five years before 2016's votes.
\newblock \emph{The Guardian}, 2017.

\bibitem[Kempe et~al.(2003)Kempe, Kleinberg, and
  Tardos]{KempeKleinbergTardos:2003}
David Kempe, Jon Kleinberg, and {\'E}va Tardos.
\newblock Maximizing the spread of influence through a social network.
\newblock In \emph{KDD 2003}, pages 137--146, 2003.

\bibitem[Kim(2011)]{Kim:2011}
Yonghwan Kim.
\newblock The contribution of social network sites to exposure to political
  difference: The relationships among snss, online political messaging, and
  exposure to cross-cutting perspectives.
\newblock \emph{Computers in Human Behavior}, 27\penalty0 (2):\penalty0
  971--977, 2011.

\bibitem[Layman et~al.(2006)Layman, Carsey, and
  Horowitz]{LaymanCarseyHorowitz:2006}
Geoffrey~C. Layman, Thomas~M. Carsey, and Juliana~Menasce Horowitz.
\newblock Party polarization in american politics: Characteristics, causes, and
  consequences.
\newblock \emph{Annual Review of Political Science}, 9\penalty0 (1):\penalty0
  83--110, 2006.

\bibitem[Lee et~al.(2014)Lee, Choi, Kim, and Kim]{LeeChoiKim:2014}
Jae~Kook Lee, Jihyang Choi, Cheonsoo Kim, and Yonghwan Kim.
\newblock Social media, network heterogeneity, and opinion polarization.
\newblock \emph{Journal of communication}, 64\penalty0 (4):\penalty0 702--722,
  2014.

\bibitem[Lord et~al.(1979)Lord, Ross, and Lepper]{LordRossLepper:1979}
Charles~G. Lord, Lee Ross, and Mark~R. Lepper.
\newblock Biased assimilation and attitude polarization: The effects of prior
  theories on subsequently considered evidence.
\newblock \emph{Journal of personality and social psychology}, 37\penalty0
  (11):\penalty0 2098, 1979.

\bibitem[Mallmann-Trenn et~al.(2018)Mallmann-Trenn, Musco, and
  Musco]{Mallmann-TrennMuscoMusco:2018}
Frederik Mallmann-Trenn, Cameron Musco, and Christopher Musco.
\newblock {Eigenvector Computation and Community Detection in Asynchronous
  Gossip Models}.
\newblock In \emph{\ICALP{2018}}, pages 159:1--159:14, 2018.

\bibitem[Matakos et~al.(2017)Matakos, Terzi, and
  Tsaparas]{MatakosTerziTsaparas2017}
Antonis Matakos, Evimaria Terzi, and Panayiotis Tsaparas.
\newblock Measuring and moderating opinion polarization in social networks.
\newblock \emph{Data Min. Knowl. Discov.}, 31\penalty0 (5):\penalty0
  1480--1505, September 2017.
\newblock ISSN 1384-5810.
\newblock \doi{10.1007/s10618-017-0527-9}.
\newblock URL \url{https://doi.org/10.1007/s10618-017-0527-9}.

\bibitem[McCright and Dunlap(2011)]{McCrightDunlap:2011}
Aaron~M McCright and Riley~E Dunlap.
\newblock The politicization of climate change and polarization in the american
  public's views of global warming, 2001--2010.
\newblock \emph{The Sociological Quarterly}, 52\penalty0 (2):\penalty0
  155--194, 2011.

\bibitem[McSherry(2001)]{McSherry:2001}
Frank McSherry.
\newblock Spectral partitioning of random graphs.
\newblock In \emph{\FOCS{2001}}, page 529, 2001.

\bibitem[Musco et~al.(2018)Musco, Musco, and
  Tsourakakis]{MuscoMuscoTsourakakis:2018}
Cameron Musco, Christopher Musco, and Charalampos Tsourakakis.
\newblock Minimizing controversy and disagreement in social networks.
\newblock 2018.
\newblock \WWW{2018}.

\bibitem[Nguyen et~al.(2014)Nguyen, Hui, Harper, Terveen, and
  Konstan]{NguyenHuiHarper:2014}
Tien~T Nguyen, Pik-Mai Hui, F~Maxwell Harper, Loren Terveen, and Joseph~A
  Konstan.
\newblock Exploring the filter bubble: the effect of using recommender systems
  on content diversity.
\newblock In \emph{\WWW{2014}}, pages 677--686. ACM, 2014.

\bibitem[Pariser(2011)]{Pariser:2011}
Eli Pariser.
\newblock \emph{The filter bubble: What the Internet is hiding from you}.
\newblock Penguin UK, 2011.

\bibitem[Riley and Robertson(2017)]{RileyRobertson2017}
Michael Riley and Jordan Robertson.
\newblock Russian cyber hacks on us electoral system far wider than previously
  known.
\newblock \emph{Bloomberg, June}, 13, 2017.

\bibitem[Shearer and Matsa(2018)]{ShearerMatsa:2018}
Elisa Shearer and Katerina~Eva Matsa.
\newblock News use across social media platforms: Most americans continue to
  get news on social media, even though many have concerns about its accuracy.
\newblock \emph{Pew Research Center Report}, 2018.
\newblock URL
  \url{https://www.journalism.org/2018/09/10/news-use-across-social-media-platforms-2018/}.

\bibitem[{Shen} et~al.(2016){Shen}, {Diamond}, {Udell}, {Gu}, and
  {Boyd}]{ShenDiamondUdellGuBoyd2016}
Xinyue {Shen}, Steven {Diamond}, Madeleine {Udell}, Yuantao {Gu}, and Stephen
  {Boyd}.
\newblock {Disciplined Multi-Convex Programming}.
\newblock \emph{arXiv e-prints}, art. arXiv:1609.03285, Sep 2016.

\bibitem[Smith and Andersen(2018)]{SmithAndersen2018}
Aaron Smith and Monica Andersen.
\newblock Social media use in 2018.
\newblock \emph{Pew Research Center Report}, 2018.

\bibitem[Smith and Christakis(2008)]{SmithChristakis:2008}
Kirsten~P Smith and Nicholas~A Christakis.
\newblock Social networks and health.
\newblock \emph{Annu. Rev. Sociol}, 34:\penalty0 405--429, 2008.

\bibitem[Spielman(2015)]{Spielman:2015}
Daniel Spielman.
\newblock Lecture notes on spectral partitioning in a stochastic block model.
\newblock http://www.cs.yale.edu/homes/spielman/561/lect21-15.pdf., 2015.

\bibitem[Vaccari et~al.(2016)Vaccari, Valeriani, Barber{\'a}, Jost, Nagler, and
  Tucker]{VaccariValerianiBarberaJostNagler2016}
Cristian Vaccari, Augusto Valeriani, Pablo Barber{\'a}, John~T. Jost, Jonathan
  Nagler, and Joshua~A. Tucker.
\newblock Of echo chambers and contrarian clubs: Exposure to political
  disagreement among german and italian users of twitter.
\newblock \emph{Social Media + Society}, 2\penalty0 (3):\penalty0
  2056305116664221, 2016.
\newblock \doi{10.1177/2056305116664221}.
\newblock URL \url{https://doi.org/10.1177/2056305116664221}.

\bibitem[Vu(2007)]{Vu:2007}
Van~H. Vu.
\newblock Spectral norm of random matrices.
\newblock \emph{Combinatorica}, 27\penalty0 (6):\penalty0 721--736, 2007.

\bibitem[Wainwright(2019)]{Wainwright:2019}
Martin~J. Wainwright.
\newblock \emph{High-Dimensional Statistics: A Non-Asymptotic Viewpoint}.
\newblock Cambridge Series in Statistical and Probabilistic Mathematics.
  Cambridge University Press, 2019.

\bibitem[Zuiderveen~Borgesius et~al.(2016)Zuiderveen~Borgesius, Trilling,
  Moeller, Bod{\'o}, de~Vreese, and
  Helberger]{Zuiderveen-BorgesiusTrillingMoeller:2016}
Frederik Zuiderveen~Borgesius, Damian Trilling, Judith Moeller, Bal{\'a}zs
  Bod{\'o}, Claes~H. de~Vreese, and Natali Helberger.
\newblock Should we worry about filter bubbles?
\newblock \emph{Internet Policy Review, Journal on Internet Regulation},
  5\penalty0 (1), 2016.

\end{thebibliography}
	\bibliographystyle{plainnat}
\end{document}